\documentclass{article} % For LaTeX2e
\usepackage{iclr2027_conference,times}

\usepackage{amsmath,amsfonts,bm}

\def\eqref#1{equation~\ref{#1}}
\def\plaineqref#1{\ref{#1}}
\def\1{\bm{1}}

\DeclareMathAlphabet{\mathsfit}{\encodingdefault}{\sfdefault}{m}{sl}
\SetMathAlphabet{\mathsfit}{bold}{\encodingdefault}{\sfdefault}{bx}{n}

\usepackage{hyperref}
\hypersetup{hypertexnames=false}
\usepackage{url}
\usepackage{graphicx}

\usepackage{amsmath}
\usepackage{amssymb}
\usepackage{amsthm}
\usepackage{algorithm}
\usepackage{algpseudocode}
\usepackage{booktabs}

\newtheorem{theorem}{Theorem}
\newtheorem{lemma}{Lemma}

\newtheorem{corollary}{Corollary}

\title{Approximation Algorithms for Dual-Memory Request Scheduling in Disaggregated LLM Inference}

\author{\parbox{0.96\textwidth}{\raggedright
Xin Wang$^{a}$, Hong Shen$^{a}$, Hui Tian$^{b}$, Ye Tao$^{a}$\\[0.4em]
{\normalfont\small $^{a}$School of Engineering and Technology, Central Queensland University, Australia}\\
{\normalfont\small $^{b}$School of Information and Communication Technology, Griffith University, Australia}
}}

\iclrfinalcopy % Show authors and suppress review line numbers in this preprint.
\begin{document}
	\maketitle
\lhead{Preprint}
\begin{abstract}
	Disaggregated LLM inference executes the prefill and decode phases on separate resource pools, reducing interference between the two phases. However, this architecture introduces a coupled scheduling problem because a request that completes prefill retains its prompt KV cache in prefill memory until it is admitted to the decode pool. Prefill and decode admission decisions therefore jointly determine memory occupancy and request completion times. To the best of our knowledge, this problem lacks a formal formulation, and no approximation guarantees are known for request scheduling in disaggregated LLM inference. We formalize this setting as the dual-memory request scheduling (DMRS) problem, whose objective is to minimize the sum of request completion times, and prove that it is strongly NP-hard. We develop two deterministic polynomial-time approximation algorithms. The first applies to instances satisfying \(\alpha<1\) and \(\beta<1\) and achieves an approximation ratio of \(8+\frac{4}{1-\alpha}+\frac{8}{1-\beta}\), where \(\alpha\) and \(\beta\) are the maximum normalized prefill and peak decode memory demands among all requests, respectively. In particular, it yields a \(32\)-approximation when
	\(\alpha,\beta\leq 1/2\). The second algorithm applies to all feasible instances and achieves a uniform \(36\)-approximation. Finally, evaluations on two H100 GPUs using trace-derived offline workloads show that the proposed algorithms reduce the sum of request completion times relative to multiple vLLM deployment baselines.
\end{abstract}
	
\section{Introduction}
\label{sec:introduction}
Large language model (LLM) inference consists of a prefill phase followed
by an autoregressive decode phase. During prefill, the model processes the prompt
tokens in parallel and constructs the key-value (KV) cache. During decoding,
output tokens are generated sequentially, and the KV cache expands with each
token. Prefill is generally compute-bound,
whereas decoding is more sensitive to memory bandwidth and KV cache
capacity~\citep{splitwise,distserve}.  When colocated on the same GPUs, the two phases compete for shared resources. To reduce this interference, recent
systems execute the prefill and decode phases on separate GPU instances or
resource pools~\citep{splitwise,distserve,tetriInfer,mooncake,hexgen2}. This
architecture allows resources to be allocated independently to the two
phases.

However, prefill-decode disaggregation does not eliminate the scheduling
dependence between the two phases. After completing prefill, a request may be unable to enter
decoding immediately if the available decode memory is insufficient. While
waiting, the request retains its prompt KV cache in prefill memory, prolonging
its occupancy of that memory and potentially blocking subsequent prefill
admissions. Upon decode admission, the KV cache is transferred to decode
memory, and the retained prefill memory is released. The request's decode
memory usage then grows as output tokens are generated. Consequently, prefill
and decode admission decisions are coupled through their effects on memory
occupancy in the two resource pools.

Existing systems work has investigated serving architectures and resource
allocation~\citep{distserve,splitwise,tetriInfer,hexgen2}, request scheduling
and load balancing to improve goodput and meet service-level
objectives~\citep{taichi,wu2025arrow,pdserve}, and KV cache transfer, storage,
and management~\citep{mooncake,memserve,smartgen,hong2025semi}. These studies
focus primarily on the design and implementation of practical serving
systems. Related theoretical work has analyzed latency-oriented request
scheduling subject to a unified memory capacity
constraint~\citep{jaillet2025online,ao2025fluid,wang2025llm,feng2026competitive}.

To the best of our knowledge, no prior theoretical scheduling model jointly
considers prefill and decode admission decisions under separate memory
capacities, and no approximation guarantees are known for minimizing the sum
of request completion times in this setting. We address this gap by studying offline
scheduling of a batch of requests under prefill-decode disaggregation. Our main contributions are summarized as follows.

\noindent\textbf{Contribution 1:}
We develop a system model for prefill and decode admission scheduling
in disaggregated LLM inference. Based on this model, we formulate the
dual-memory request scheduling (DMRS) problem with the objective of
minimizing the sum of request completion times. We prove that DMRS is
strongly NP-hard.

\noindent\textbf{Contribution 2:}
We develop a deterministic polynomial-time approximation algorithm for
instances satisfying \(\alpha<1\) and \(\beta<1\). It achieves an
approximation ratio of
\(8+\frac{4}{1-\alpha}+\frac{8}{1-\beta}\),
where \(\alpha\) and \(\beta\) are the maximum normalized prefill and peak
decode memory demands, respectively. In particular, it yields a
\(32\)-approximation when \(\alpha,\beta\leq1/2\).

\noindent\textbf{Contribution 3:}
We develop a second deterministic polynomial-time algorithm that applies
to every feasible instance, including \(\alpha=1\) or \(\beta=1\), and achieves a uniform \(36\)-approximation.

\noindent\textbf{Contribution 4:}
We evaluate both algorithms on synthetic instances and workloads derived
from the Azure LLM Inference Trace 2023~\citep{splitwise}. The CPU experiments compare the algorithms' objective values with the exact
optimum obtained from a time-indexed integer linear programming (ILP)
formulation and with a linear programming (LP) lower bound.
Serving experiments on two H100 GPUs use Qwen2.5
(7B)~\citep{yang2024qwen25} and Llama~3.1
(8B)~\citep{grattafiori2024llama} to compare total completion time against
three vLLM baselines.

\section{System Model and Formulation}
\label{model_formulation}
In this section, we first present the system model. We then formulate DMRS,
establish its computational complexity, and derive an LP relaxation for our approximation analysis. Appendix~\ref{app:notation} summarizes the
notation used throughout the paper.

\subsection{System Model}
\label{sec:system_model}
In our model, the prefill and decode phases execute on separate
resource pools with respective memory capacities
\(M_P,M_D\in\mathbb Z_{>0}\).

 \noindent\textbf{Execution Setup.}
Time is divided into slots \([t,t+1)\), where
\(t\in\{0,1,2,\ldots\}\). KV cache memory is measured in token units, with one
unit representing the memory required to store the key and value states of one
token. During prefill, each active request processes its prompt for a fixed number
of slots. During decoding, each active request completes one decode
iteration and generates one output token per slot. Multiple requests may execute
concurrently in each phase, subject to the corresponding memory capacity.
Execution is non-preemptive within each phase: once admitted to a phase, a
request remains active until it completes that phase.

\noindent\textbf{Inference Requests.}
Let \(\mathcal I=\{1,\ldots,n\}\) denote the set of requests. Each request
\(i\in\mathcal I\) is characterized by
\((m_i,p_i,d_i)\in\mathbb Z_{>0}^3\), where \(m_i\) is its effective prefill
processing time, measured in advance under the given deployment configuration
and expressed in slots; \(p_i\) is the number of input tokens in its prompt
and, equivalently, its prefill memory demand in token units; and \(d_i\) is
the number of output tokens to be generated and, equivalently, the number of
decode iterations. After generating \(u\) output tokens, request \(i\) has
decode memory demand \(r_i(u)=p_i+u\). Its peak decode memory demand is
\(h_i=r_i(d_i)=p_i+d_i\).

Throughout the paper, we assume that \(p_i\leq M_P\) and \(h_i\leq M_D\) for
every \(i\in\mathcal I\). These conditions are necessary and sufficient for feasibility, since the
requests can always be processed sequentially. Let \(\alpha:=\max_{i\in\mathcal I}p_i/M_P\) and
\(\beta:=\max_{i\in\mathcal I}h_i/M_D\) denote the maximum normalized
prefill and peak decode memory demands, respectively. The feasibility
conditions imply \(\alpha,\beta\leq1\).

\noindent\textbf{System Overview.}
As shown in Figure~\ref{fig:system_model}, we consider a batch of requests released at time \(t=0\), with their characteristics
\((m_i,p_i,d_i)\) known to the scheduler in advance. During execution, each unfinished request belongs to one of four state sets.

\begin{figure}[htp]
	\centering
	\includegraphics[width=0.8\linewidth]{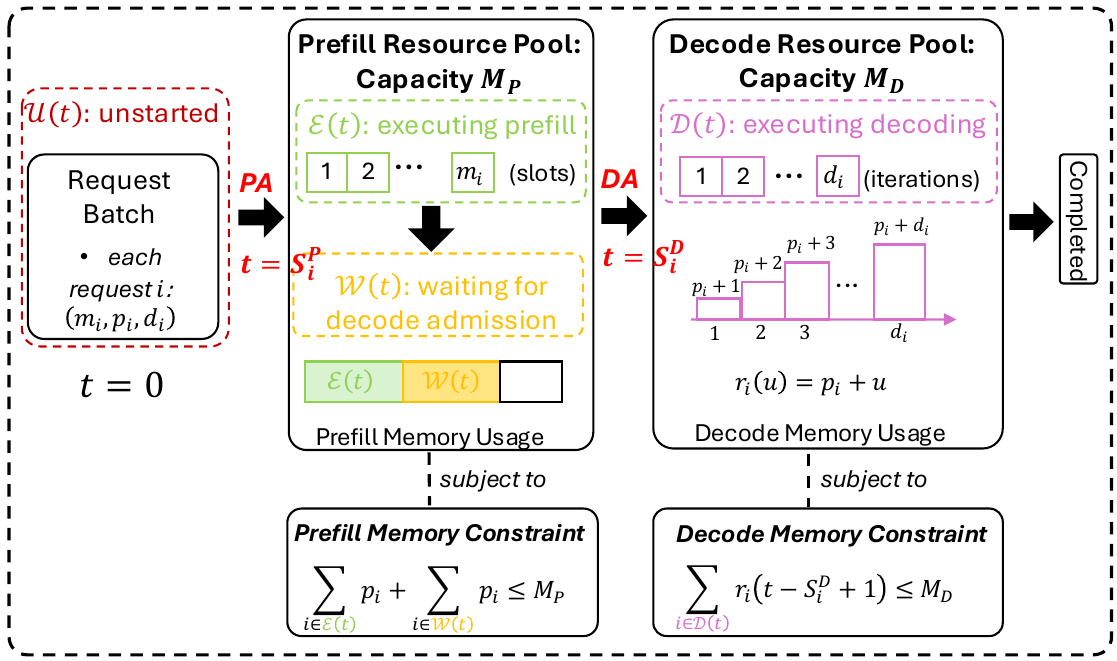}
	\caption{System model for joint prefill and decode admission scheduling.}
	\label{fig:system_model}
\end{figure}

For each slot \([t,t+1)\), let \(\mathcal U(t)\), \(\mathcal E(t)\),
\(\mathcal W(t)\), and \(\mathcal D(t)\) denote, respectively, the sets of
unstarted requests awaiting prefill admission (PA), requests executing
prefill, requests awaiting decode admission (DA), and requests executing
decoding. Let \(S_i^P\) and \(S_i^D\) denote the prefill and decode admission
times of request \(i\), respectively.

\textbf{PA and Prefill Memory Constraint.}
The scheduler selects a feasible subset of unstarted requests for prefill
admission. A request \(i\) admitted at time \(S_i^P\) enters the prefill
pool and belongs to \(\mathcal E(t)\) while executing prefill. It requires
\(m_i\) slots to process prompt tokens in parallel and generate its
prompt KV cache, and it occupies \(p_i\) units of prefill memory. Requests
not admitted remain in \(\mathcal U(t)\).

After completing prefill, request \(i\) is either admitted immediately to
decoding or enters \(\mathcal W(t)\) while awaiting decode admission.
During this waiting period, it retains its prompt KV cache and continues
to occupy \(p_i\) units of prefill memory. Therefore, request \(i\)
occupies \(p_i\) units of prefill memory throughout
\([S_i^P,S_i^D)\). The prefill memory constraint during slot
\([t,t+1)\) is thus
\[
\sum_{i\in\mathcal E(t)} p_i
+
\sum_{i\in\mathcal W(t)} p_i
\leq M_P.
\]
\textbf{DA and Decode Memory Constraint.}
At admission time \(S_i^D\), request \(i\) enters the decode
pool. Its prompt KV cache is transferred from prefill memory to decode
memory, and its prefill memory is released. We model this transfer as
instantaneous and omit transfer time and network bandwidth constraints.
During decoding, request \(i\) belongs to \(\mathcal D(t)\) and completes
one iteration per slot.

After generating \(u\) output tokens, request \(i\) has decode memory
demand \(r_i(u)=p_i+u\). For \(u=0,\ldots,d_i\), this demand grows from
\(p_i\) at decode admission to the peak value
\(h_i=r_i(d_i)=p_i+d_i\). If request \(i\) executes decoding during slot
\([t,t+1)\), it has generated \(t-S_i^D\) output tokens before the slot
and generates one additional token during the slot. Its decode memory
demand therefore reaches \(r_i(t-S_i^D+1)\) during that slot. Hence, the
decode memory constraint is
\[
\sum_{i\in\mathcal D(t)}(p_i+t-S_i^D+1)\leq M_D.
\]
After its \(d_i\)-th decode iteration, request \(i\) leaves the system and
releases its decode memory.

\subsection{DMRS Formulation}
\label{sec:problem_formulation}
Based on the system model above, we formulate joint prefill and decode
admission scheduling as the dual-memory request scheduling (DMRS)
problem, whose objective is to minimize the sum of request completion times.
The optimization problem, denoted by \(\mathbf{(O)}\), is formulated
as follows:
\begin{subequations}
	\label{eq:p0}
	\begin{align}
		\min\quad &
		\sum_{i\in\mathcal I} C_i
		\label{eq:p0_obj}\\
		\text{s.t.}\quad &
		C_i^P=S_i^P+m_i,
		\quad \forall i\in\mathcal I,
		\label{eq:p0_1}\\
		&
		C_i=S_i^D+d_i,
		\quad \forall i\in\mathcal I,
		\label{eq:p0_2}\\
		&
		S_i^D\geq C_i^P,
		\quad \forall i\in\mathcal I,
		\label{eq:p0_3}\\
		&
		\sum_{i\in\mathcal I}
		p_i\,\mathbf{1}\!\left\{S_i^P\leq t<S_i^D\right\}
		\leq M_P,
		\quad \forall t\in\{0,1,2,\ldots\},
		\label{eq:p0_4}\\
		&
		\sum_{i\in\mathcal I}
		\left(p_i+t-S_i^D+1\right)
		\mathbf{1}\!\left\{S_i^D\leq t<C_i\right\}
		\leq M_D,
		\quad \forall t\in\{0,1,2,\ldots\},
		\label{eq:p0_5}\\
		&
		S_i^P,S_i^D\in\mathbb Z_{\geq 0},
		\quad \forall i\in\mathcal I.
		\label{eq:p0_6}
	\end{align}
\end{subequations}
For each request \(i\in\mathcal I\), the scheduler determines its prefill
and decode admission times \(S_i^P\) and \(S_i^D\).
Constraints~(\plaineqref{eq:p0_1}) and~(\plaineqref{eq:p0_2}) define its
prefill completion time and request completion time, respectively.
Constraint~(\plaineqref{eq:p0_3}) ensures that request \(i\) is admitted to
decoding after completing prefill. Constraints~(\plaineqref{eq:p0_4})
and~(\plaineqref{eq:p0_5}) enforce the prefill and decode memory capacities,
respectively, where \(\mathbf{1}\{\cdot\}\) denotes the indicator function.
Constraint~(\plaineqref{eq:p0_6}) restricts admission times to nonnegative
integers.

Since all requests are released at time zero, \(C_i\) is also the end-to-end
latency of request \(i\), i.e., minimizing
\(\sum_{i\in\mathcal I}C_i\) is equivalent to minimizing the average
end-to-end latency. Let \(\mathrm{OPT}\) denote the optimal objective value
of \(\mathbf{(O)}\).

\begin{theorem}
	\label{thm:hardness}
	DMRS is strongly NP-hard, even when every request requires exactly one
	slot of prefill processing and one decode iteration. (The complete proof is given in
	Appendix~\ref{app:proof:thm-hardness}.)
\end{theorem}
	
\subsection{LP Relaxation}
\label{sec:relaxation}
For the approximation analysis, we derive an LP relaxation from necessary
conditions satisfied by every feasible schedule. The relaxation captures
the mandatory memory-time loads imposed by each request on the prefill and
decode memory pools. Its optimal value provides a lower bound on
\(\mathrm{OPT}\).

\noindent\textbf{Prefill Memory-Time Load.}
During its \(m_i\) slots of prefill processing, request \(i\) occupies
\(p_i\) units of prefill memory. Its mandatory prefill memory-time load is
\(\ell_i^P:=p_i m_i\).

\noindent\textbf{Decode Memory-Time Load.}
During its \(u\)-th decode iteration, request \(i\) requires \(p_i+u\)
units of decode memory. Its mandatory decode memory-time load is
\(\ell_i^D
:=
\sum_{u=1}^{d_i}(p_i+u)
=
p_i d_i+\frac{d_i(d_i+1)}{2}\).

Let \(\bar C_i\) denote the relaxed completion-time variable associated with
request \(i\). We consider the following linear program, denoted by
\(\mathbf{(LP)}\):
\begin{subequations}
	\label{eq:lp}
	\begin{align}
		\min
		\quad&
		\sum_{i\in\mathcal I}\bar C_i
		\label{eq:lp_obj}\\
		\text{s.t.}\quad&
		\bar C_i\geq m_i+d_i,
		\qquad \forall i\in\mathcal I,
		\label{eq:lp_single}\\
		&
		\sum_{i\in\mathcal S}\ell_i^P\bar C_i
		\geq
		\frac{1}{2M_P}
		\left[
		\left(\sum_{i\in\mathcal S}\ell_i^P\right)^2
		+
		\sum_{i\in\mathcal S}(\ell_i^P)^2
		\right],
		\quad
		\forall\,\emptyset\ne\mathcal S\subseteq\mathcal I,
		\label{eq:lp_prefill_subset}\\
		&
		\sum_{i\in\mathcal S}\ell_i^D\bar C_i
		\geq
		\frac{1}{2M_D}
		\left[
		\left(\sum_{i\in\mathcal S}\ell_i^D\right)^2
		+
		\sum_{i\in\mathcal S}(\ell_i^D)^2
		\right],
		\quad
		\forall\,\emptyset\ne\mathcal S\subseteq\mathcal I.
		\label{eq:lp_decode_subset}
	\end{align}
\end{subequations}

Let \(\mathrm{LP}^{\star}\) denote the optimal value of
\(\mathbf{(LP)}\). Constraint~(\plaineqref{eq:lp_single}) provides a lower bound on each
request's completion time.
Constraints~(\plaineqref{eq:lp_prefill_subset}) and
(\plaineqref{eq:lp_decode_subset}) provide lower bounds derived from the
mandatory prefill and decode memory-time loads of every nonempty subset of
requests, respectively. Lemma~\ref{lem:lp_lower_bound} establishes that these
constraints are satisfied by the completion times of every feasible
schedule.

\begin{lemma}
	\label{lem:lp_lower_bound}
	For every feasible schedule of \(\mathbf{(O)}\), setting
	\(\bar C_i=C_i\) for all \(i\in\mathcal I\) yields a feasible solution to
	\(\mathbf{(LP)}\). Consequently,
	\(\mathrm{LP}^{\star}\leq\mathrm{OPT}\). (The complete proof is given in Appendix~\ref{app:proof:lem-lp_lower_bound}.)
\end{lemma}
	
\begin{lemma}
	\label{lem:lp_polynomial}
	The linear program \(\mathbf{(LP)}\) can be solved in polynomial time. (The complete proof is given in Appendix~\ref{app:proof:lem-lp_polynomial}.)
\end{lemma}
	
\section{Approximation for Restricted Instances}
\label{sec:app_alg1}
Algorithm~\ref{alg:bounded_scheduling} applies to instances satisfying
\(\alpha<1\) and \(\beta<1\). Hence, every request satisfies
\(M_P-p_i\geq(1-\alpha)M_P\) and \(M_D-h_i\geq(1-\beta)M_D\).
If a request is rejected by the prefill admission rule, the prefill memory
occupied by requests exceeds \((1-\alpha)M_P\). Similarly, if a request
is rejected by the decode admission rule, the peak decode memory reserved by
requests exceeds \((1-\beta)M_D\).

\subsection{Algorithm Design}
Using an optimal LP solution \(\{\bar C_i\}_{i\in\mathcal I}\),
Algorithm~\ref{alg:bounded_scheduling} partitions the requests into dyadic
groups \(\mathcal G_q:=\{i\in\mathcal I:
2^{q-1}<\bar C_i\leq 2^q\}\) for \(q\in\mathbb Z\).

\begin{algorithm}[htp]
	\caption{LP-Guided Dyadic Group Scheduling}
	\label{alg:bounded_scheduling}
	\begin{algorithmic}[1]
		\Require Requests \(\{(m_i,p_i,d_i)\}_{i\in\mathcal I}\) and memory
		capacities \(M_P\) and \(M_D\)
		\Ensure Admission times
		\(\{(S_i^P,S_i^D)\}_{i\in\mathcal I}\)

		\State Solve \(\mathbf{(LP)}\) and obtain an optimal solution
		\(\{\bar C_i\}_{i\in\mathcal I}\); Construct all nonempty dyadic groups \(\mathcal G_q\)
		\State \(t\gets0\);
		\(\mathcal E,\mathcal W,\mathcal D\gets\emptyset\)

		\For{each nonempty \(\mathcal G_q\), in increasing order of \(q\)}
		\State \(\mathcal U_q\gets\mathcal G_q\)

		\While{\(\mathcal U_q\cup\mathcal E\cup\mathcal W\cup\mathcal D
			\neq\emptyset\)}

		\State Remove every request \(i\in\mathcal D\) satisfying
		\(S_i^D+d_i=t\)
		\State Move every request \(i\in\mathcal E\) satisfying
		\(S_i^P+m_i=t\) to \(\mathcal W\)

		% \Statex \quad \(\triangleright\) \textit{Decode Admission with Peak Reservation (Lines 8--13)}
		\State \(\widehat{\mathcal W}\gets\mathcal W\)
		\For{each \(i\in\widehat{\mathcal W}\), in nondecreasing order of
			\(m_i+d_i\)}
		\If{\(\sum_{j\in\mathcal D}h_j+h_i\leq M_D\)}
		\Comment{reserve peak decode demands}
		\State \(S_i^D\gets t\); \(\mathcal W\gets\mathcal W\setminus\{i\}\);
		\(\mathcal D\gets\mathcal D\cup\{i\}\)
		\EndIf
		\EndFor

		% \Statex \quad \(\triangleright\) \textit{Prefill Admission Rule (Lines 14--19)}
		\State \(\widehat{\mathcal U}_q\gets\mathcal U_q\)
		\For{each \(i\in\widehat{\mathcal U}_q\), in nondecreasing order of
			\(m_i+d_i\)}
		\If{\(\sum_{j\in\mathcal E}p_j+
			\sum_{j\in\mathcal W}p_j+p_i\leq M_P\)}
		\State \(S_i^P\gets t\); \(\mathcal U_q\gets\mathcal U_q\setminus\{i\}\);
		\(\mathcal E\gets\mathcal E\cup\{i\}\)
		\EndIf
		\EndFor

		\If{\(\mathcal U_q\cup\mathcal E\cup\mathcal W\cup\mathcal D
			\neq\emptyset\)}
		\State \(t\gets
		\min\left(
		\{S_i^P+m_i:i\in\mathcal E\}
		\cup
		\{S_i^D+d_i:i\in\mathcal D\}
		\right)\)
		\EndIf
		\EndWhile
		\EndFor
	\end{algorithmic}
\end{algorithm}

The nonempty groups are processed in increasing order of \(q\). No request
from a later group is admitted before all requests in the current group have
completed. In each group, we order the requests by nondecreasing \(m_i+d_i\), with ties broken by
request index.

At each event time, Algorithm~\ref{alg:bounded_scheduling} performs decode
admission before prefill admission. It reserves \(h_i\) for each request
executing decoding until completion, while its actual decode memory demand
satisfies \(r_i(u)\leq h_i\) for \(0\leq u\leq d_i\). Appendix~\ref{app:alg1-pseudocode} provides a line-by-line
description of Algorithm~\ref{alg:bounded_scheduling}.

\subsection{Approximation Guarantee}
The analysis partitions the slots before each dyadic prefix completes into
four cases: a request remains waiting for decode admission, a request remains
unadmitted to prefill, or the current group is in its final prefill or decode
tail. Failed admissions yield lower bounds on reserved decode memory or
occupied prefill memory, allowing the LP prefix load bounds to control the
first two cases; the reserved decode memory-time load is less than twice the
mandatory decode load. The dyadic group bounds control the two tails. The
lemmas and charging argument appear in Appendix~\ref{app:alg1-proofs}.

\begin{theorem}
	\label{thm:alg1_ratio}
	For every instance satisfying \(\alpha<1\) and \(\beta<1\),
	Algorithm~\ref{alg:bounded_scheduling} is a deterministic polynomial-time
	\(\left(8+\frac{4}{1-\alpha}+\frac{8}{1-\beta}\right)\)-approximation.
\end{theorem}

\begin{proof}
	For each \(k\in\mathbb Z\), define the dyadic prefix
	\(\mathcal G_{\leq k}:=\bigcup_{q\leq k}\mathcal G_q\).
	For \(k\) such that \(\mathcal G_{\leq k}\neq\emptyset\), let \(T_k\)
	denote the time at which requests in \(\mathcal G_{\leq k}\) have
	completed.
	Lemma~\ref{lem:alg1_prefix_load} bounds the mandatory prefill and decode
	memory-time loads of this prefix by \(2M_P2^k\) and \(2M_D2^k\),
	respectively. Lemma~\ref{lem:alg1_decode_reservation} shows that the reserved decode
	memory-time load is less than twice the mandatory decode load. Using
	these bounds, the charging argument in
	Lemma~\ref{lem:alg1_prefix_completion} shows that requests in
	\(\mathcal G_{\leq k}\) complete by time
	\[
	T_k
	\leq
	\left(
	4+\frac{2}{1-\alpha}+\frac{4}{1-\beta}
	\right)2^k.
	\]

	For any request \(i\in\mathcal G_k\), the group definition gives
	\(2^{k-1}<\bar C_i\), and hence \(2^k<2\bar C_i\). Since \(C_i\leq T_k\),
	\[
	C_i
	\leq
	\left(
	8+\frac{4}{1-\alpha}+\frac{8}{1-\beta}
	\right)\bar C_i.
	\]
	Summing over all requests and applying
	Lemma~\ref{lem:lp_lower_bound}, which gives
	\(\sum_i\bar C_i=\mathrm{LP}^{\star}\leq\mathrm{OPT}\), proves the stated
	approximation ratio.

	By Lemma~\ref{lem:lp_polynomial}, \(\mathbf{(LP)}\) can be solved in
	polynomial time. The grouping, sorting, and event-driven scheduling
	are also polynomial. Hence,
	Algorithm~\ref{alg:bounded_scheduling} is deterministic and runs in
	polynomial time.
\end{proof}

\begin{corollary}
	\label{cor:alg1_half_capacity}
If \(\alpha,\beta\leq1/2\), equivalently, if \(p_i\leq M_P/2\) and
\(h_i\leq M_D/2\) for every \(i\in\mathcal I\), then
Algorithm~\ref{alg:bounded_scheduling} is a \(32\)-approximation.
More generally, the approximation bound tends to \(20\) as
\(\alpha,\beta\to0\) and diverges as either parameter approaches \(1\)
from below.
\end{corollary}

\section{Approximation for General Instances}
\label{sec:alg2}
To obtain a uniform guarantee, we develop a \(36\)-approximation for every
feasible instance, including the boundary cases \(\alpha=1\) or \(\beta=1\).
Algorithm~\ref{alg:general_scheduling} combines a classification based on
prefill and peak decode memory demands with geometric grouping based on an
optimal LP solution.

\subsection{Request Classification}
We classify each request by its prefill and peak decode memory demands:
\[
\begin{aligned}
	\mathcal L_P
	&:=
	\left\{i\in\mathcal I:p_i>\frac{M_P}{2}\right\},
	&
	\mathcal S_P
	&:=
	\left\{i\in\mathcal I:p_i\leq\frac{M_P}{2}\right\},\\
	\mathcal L_D
	&:=
	\left\{i\in\mathcal I:h_i>\frac{M_D}{2}\right\},
	&
	\mathcal S_D
	&:=
	\left\{i\in\mathcal I:h_i\leq\frac{M_D}{2}\right\}.
\end{aligned}
\]
Requests in \(\mathcal L_P\) and \(\mathcal S_P\) are called
\emph{prefill-large} and \emph{prefill-small}, while those in
\(\mathcal L_D\) and \(\mathcal S_D\) are called \emph{decode-large} and
\emph{decode-small}, respectively.

The sets \(\mathcal C^{LL}:=\mathcal L_P\cap\mathcal L_D\),
\(\mathcal C^{SL}:=\mathcal S_P\cap\mathcal L_D\),
\(\mathcal C^{LS}:=\mathcal L_P\cap\mathcal S_D\), and
\(\mathcal C^{SS}:=\mathcal S_P\cap\mathcal S_D\) partition
\(\mathcal I\). In each superscript, the first and second letters indicate
the prefill and decode classifications, respectively. The strict \(1/2\)
threshold ensures that at most one prefill-large request belongs to
\(\mathcal E(t)\cup\mathcal W(t)\) and, under peak decode reservation,
at most one decode-large request belongs to \(\mathcal D(t)\).

\subsection{Algorithm Design}
Fix a geometric grouping ratio \(\rho>1\), let \(B_q=\rho^q\), and define
\(\mathcal H_q:=\{i\in\mathcal I:
B_q/\rho<\bar C_i\leq B_q\}\) for \(q\in\mathbb Z\).

Algorithm~\ref{alg:general_scheduling} processes the nonempty groups in
increasing order of \(q\), completing each group before starting the next.
Within each group \(\mathcal H_q\), it processes
\(\mathcal H_q\cap\mathcal C^{LL}\),
\(\mathcal H_q\cap\mathcal C^{SL}\),
\(\mathcal H_q\cap\mathcal C^{LS}\), and
\(\mathcal H_q\cap\mathcal C^{SS}\) in this order. The first subset is
processed sequentially, while the other three use event-driven admission
subject to both memory capacities. Within each subset, requests are ordered
by nondecreasing \(m_i+d_i\), with ties broken by request index.
For decode admission, the algorithm reserves \(h_j\) for every request
\(j\in\mathcal D\) until that request completes. Thus, decode admission
requires \(\sum_{j\in\mathcal D}h_j+h_i\leq M_D\) and is not based on the
current decode memory usage. Appendix~\ref{app:alg2-pseudocode} provides a detailed line-by-line
description of Algorithm~\ref{alg:general_scheduling}.

\begin{algorithm}[htp]
	\caption{LP-Guided Geometric Group Scheduling}
	\label{alg:general_scheduling}
	\begin{algorithmic}[1]
		\Require Requests \(\{(m_i,p_i,d_i)\}_{i\in\mathcal I}\), memory capacities
		\(M_P,M_D\), and grouping ratio \(\rho>1\)
		\Ensure Admission times
		\(\{(S_i^P,S_i^D)\}_{i\in\mathcal I}\)

		\State Solve \(\mathbf{(LP)}\) and obtain an optimal solution
		\(\{\bar C_i\}_{i\in\mathcal I}\); Construct the classes
		\(\mathcal C^{LL},\mathcal C^{SL},\mathcal C^{LS},\mathcal C^{SS}\)
		\State Construct all nonempty geometric groups \(\mathcal H_q\)
		\State \(t\gets0\);
		\(\mathcal E,\mathcal W,\mathcal D\gets\emptyset\)

		\For{each nonempty \(\mathcal H_q\), in increasing order of \(q\)}

		\For{each \(i\in\mathcal H_q\cap\mathcal C^{LL}\), in nondecreasing
			order of \(m_i+d_i\)}
		\State \(S_i^P\gets t\);
		\(S_i^D\gets t+m_i\); \(t\gets t+m_i+d_i\)
		\EndFor

		\For{each \(\mathcal C\) in
			\((\mathcal C^{SL},\mathcal C^{LS},\mathcal C^{SS})\)}
		\State
		\(\mathcal U_{q,\mathcal C}
		\gets\mathcal H_q\cap\mathcal C\)

		\While{
			\(\mathcal U_{q,\mathcal C}
			\cup\mathcal E\cup\mathcal W\cup\mathcal D
			\neq\emptyset\)}

		\State Remove every \(i\in\mathcal D\) satisfying
		\(S_i^D+d_i=t\)
		\State Move every \(i\in\mathcal E\) satisfying
		\(S_i^P+m_i=t\) to \(\mathcal W\)

		%\Statex\qquad \(\triangleright\)\ \textit{Decode Admission with Peak Reservation (Lines~13--18)}
		\State \(\widehat{\mathcal W}\gets\mathcal W\)
		\For{each \(i\in\widehat{\mathcal W}\), in nondecreasing
			order of \(m_i+d_i\)}
		\If{\(\sum_{j\in\mathcal D}h_j+h_i\leq M_D\)}
		\Comment{reserve peak decode demands}
		\State \(S_i^D\gets t\);
		\(\mathcal W\gets\mathcal W\setminus\{i\}\);
		\(\mathcal D\gets\mathcal D\cup\{i\}\)
		\EndIf
		\EndFor

		%\Statex\qquad \(\triangleright\)\  \textit{Prefill Admission Rule (Lines~19--24)}
		\State
		\(\widehat{\mathcal U}_{q,\mathcal C}
		\gets\mathcal U_{q,\mathcal C}\)
		\For{each \(i\in\widehat{\mathcal U}_{q,\mathcal C}\),
			in nondecreasing order of \(m_i+d_i\)}
		\If{
			\(\sum_{j\in\mathcal E}p_j
			+\sum_{j\in\mathcal W}p_j+p_i
			\leq M_P\)}
		\State \(S_i^P\gets t\); \(\mathcal U_{q,\mathcal C}
		\gets\mathcal U_{q,\mathcal C}\setminus\{i\}\);
		\(\mathcal E\gets\mathcal E\cup\{i\}\)
		\EndIf
		\EndFor

		\If{
			\(\mathcal U_{q,\mathcal C}
			\cup\mathcal E\cup\mathcal W\cup\mathcal D
			\neq\emptyset\)}
		\State
		\(t\gets
		\min\bigl(
		\{S_i^P+m_i:i\in\mathcal E\}
		\cup
		\{S_i^D+d_i:i\in\mathcal D\}
		\bigr)\)
		\EndIf

		\EndWhile
		\EndFor
		\EndFor
	\end{algorithmic}
\end{algorithm}

\subsection{Approximation Guarantee}
The proof assigns each execution slot to one of three types. The slot either
processes a request that is large in the relevant memory pool, occurs while a
small request is blocked by the corresponding admission constraint, or
belongs to the final prefill or decode tail of a class phase. The LP bounds
on the cumulative loads of each geometric prefix, together with the peak
reservation bound, control the first two types. Geometric grouping bounds the
accumulated tails. 

Appendix~\ref{app:alg2-proofs} gives the complete slot
classification and charging argument.

\begin{theorem}
	\label{thm:alg2_ratio}
	With \(\rho=3/2\), Algorithm~\ref{alg:general_scheduling} is a
	deterministic polynomial-time \(36\)-approximation for every feasible DMRS
	instance.
\end{theorem}

\begin{proof}
	For any nonempty group \(\mathcal H_k\), define
	\(\mathcal H_{\leq k}:=\bigcup_{q\leq k}\mathcal H_q\) and let
	\(T_k:=\max_{i\in\mathcal H_{\leq k}}C_i\).
	Lemma~\ref{lem:alg2_prefix_load} bounds the mandatory prefill and decode
	memory-time loads of this prefix by \(2M_PB_k\) and \(2M_DB_k\),
	respectively. Using these bounds, the charging argument in
	Lemma~\ref{lem:alg2_prefix_completion} gives
	\(T_k
	\leq
	\left(12+\frac{4\rho}{\rho-1}\right)B_k\). 
	
	For every \(i\in\mathcal H_k\), the group definition gives
	\(B_k<\rho\bar C_i\), while \(C_i\leq T_k\). Therefore,
	\(C_i
	<
	\left(12\rho+\frac{4\rho^2}{\rho-1}\right)\bar C_i\).
	Let
	\(	f(\rho):=12\rho+\frac{4\rho^2}{\rho-1}\). For every \(\rho>1\),
	\(f(\rho)-36
	=
	\frac{4(2\rho-3)^2}{\rho-1}
	\geq0\),
	with equality at \(\rho=3/2\). Hence, by
	Lemma~\ref{lem:lp_lower_bound},
	\(	\sum_{i\in\mathcal I}C_i
	<
	36\sum_{i\in\mathcal I}\bar C_i
	=
	36\,\mathrm{LP}^{\star}
	\leq
	36\,\mathrm{OPT}\).
	
Lemma~\ref{lem:lp_polynomial}, together with the deterministic
polynomial-time grouping and scheduling steps, implies that
Algorithm~\ref{alg:general_scheduling} is deterministic and
polynomial-time.
\end{proof}

\section{Experiments}
\label{sec:experiments}
\noindent\textbf{CPU Evaluation.}
We compare the objective values of Algorithms~\ref{alg:bounded_scheduling}
and~\ref{alg:general_scheduling} with the exact optimum computed by the ILP in Appendix~\ref{app:ilp_formulation} and the LP lower bound
\(\mathrm{LP}^{\star}\) defined in Section~\ref{sec:relaxation}.

\noindent\textbf{GPU Evaluation.}
We compare the total completion time (TCT) of Algorithms~\ref{alg:bounded_scheduling}
and~\ref{alg:general_scheduling} with three vLLM
baselines: disaggregated prefill and decode serving (P/D), colocated
serving (Colocated), and colocated serving with chunked prefill
(Colocated Chunked).

\subsection{CPU Evaluation}
\label{sec:cpu_experiments}

\noindent\textbf{Setup.}
We implement Algorithms~\ref{alg:bounded_scheduling}
and~\ref{alg:general_scheduling} in Python~3.12.12 and solve the ILP and LP
using Gurobi~13.0.1 on an ARM64 macOS machine. We consider synthetic
instances and conversation and code workloads derived from the Azure LLM
Inference Trace 2023~\citep{splitwise}.
\begin{table}[htp]
	\centering
	\caption{ALG/OPT and ALG/LP for \(\kappa=2\), where
		\(M_P=\lceil \kappa\max_i p_i\rceil\) and
		\(M_D=\lceil \kappa\max_i h_i\rceil\).}
	\label{tab:cpu_results}
	\small
	\setlength{\tabcolsep}{3pt}
	\begin{tabular}{@{}lccc@{}}
		\toprule
		Workload & Scale
		& Algorithm~\ref{alg:bounded_scheduling}
		& Algorithm~\ref{alg:general_scheduling} \\
		\midrule
		Synthetic, ALG/OPT   & \(n=10\)  & \textbf{1.660} & \textbf{1.945} \\
		Synthetic, ALG/LP    & \(n=250\) & \textbf{1.749} & \textbf{1.938} \\
		Conversation, ALG/LP & \(n=250\) & \textbf{1.491} & \textbf{1.685} \\
		Code, ALG/LP         & \(n=250\) & \textbf{1.695} & \textbf{2.109} \\
		\bottomrule
	\end{tabular}
\end{table}

\noindent\textbf{Results.}
Algorithm~\ref{alg:bounded_scheduling} has a lower ratio than
Algorithm~\ref{alg:general_scheduling} in each reported setting.
All ratios are below the worst-case guarantees.
Further scaling and capacity results appear in Appendix~\ref{app:experiments}.

\subsection{GPU Evaluation}
\label{sec:gpu_experiments}
\noindent\textbf{Setup.}
We conduct serving experiments on two NVIDIA H100 GPUs using Qwen2.5
(7B)~\citep{yang2024qwen25} and Llama~3.1
(8B)~\citep{grattafiori2024llama}. We use separate offline batches of
conversation and code requests derived from the Azure LLM Inference Trace
2023~\citep{splitwise}.

\begin{table}[htp]
	\centering
	\caption{TCT normalized by the corresponding P/D result for \(n=250\) and
		\(\kappa=1.5\).}
	\label{tab:gpu_results}
	\small
	\setlength{\tabcolsep}{3pt}
	\begin{tabular}{@{}llccccc@{}}
		\toprule
		Model & Workload & P/D & Colocated
		& \shortstack{Colocated Chunked}
		& Algorithm~\ref{alg:bounded_scheduling}
		& Algorithm~\ref{alg:general_scheduling} \\
		\midrule
		Qwen2.5 (7B)    & Conversation & 1.000 & 0.759 & 1.349 & \textbf{0.787} & \textbf{0.878} \\
		Qwen2.5 (7B)    & Code         & 1.000 & 0.840 & 1.259 & \textbf{0.525} & \textbf{0.582} \\
		Llama~3.1 (8B) & Conversation & 1.000 & 0.738 & 1.239 & \textbf{0.758} & \textbf{0.839} \\
		Llama~3.1 (8B) & Code         & 1.000 & 0.835 & 1.227 & \textbf{0.510} & \textbf{0.565} \\
		\bottomrule
	\end{tabular}
\end{table}

Algorithms~\ref{alg:bounded_scheduling} and
\ref{alg:general_scheduling}, as well as the vLLM~\citep{vllm} P/D baseline, assign one GPU to prefill and the second to decoding. Colocated runs a single vLLM instance using tensor parallelism across both GPUs, so both GPUs execute prefill and decoding. Colocated Chunked uses the same configuration with chunked prefill enabled.

\noindent\textbf{Results.}
Both algorithms reduce TCT relative to P/D for each model and workload
shown in Table~\ref{tab:gpu_results}. The reductions range from
21.3\% to 49.0\% for Algorithm~\ref{alg:bounded_scheduling} and from
12.2\% to 43.5\% for Algorithm~\ref{alg:general_scheduling}. Both
achieve larger reductions on the code workloads. Additional memory
capacities, prefill and decode demands, and batch sizes are evaluated in
Appendix~\ref{app:experiments}.

\section{Related Work}
\label{sec:related_work}
\noindent\textbf{LLM Inference Serving and KV Cache Management.}
Orca~\citep{orca} introduces iteration-level scheduling,
vLLM~\citep{vllm} manages KV caches with PagedAttention, and
Sarathi~\citep{sarathi} uses chunked prefills.
FastServe~\citep{wu2023fast} uses preemptive scheduling,
AlpaServe~\citep{li2023alpaserve} multiplexes models through model
parallelism, and FlexGen~\citep{sheng2023flexgen} offloads model state.
DeepSpeed-Inference~\citep{aminabadi2022deepspeed} optimizes parallel
execution, while DeepSpeed-FastGen~\citep{holmes2024deepspeed} uses
Dynamic SplitFuse. SGLang~\citep{zheng2024sglang},
Preble~\citep{srivatsa2025preble}, and CacheGen~\citep{liu2024cachegen}
address prefix reuse, distributed scheduling, and KV cache compression
and streaming, respectively. These systems do not model joint admission
under separate prefill and decode memory capacities.

\noindent\textbf{Prefill-Decode Disaggregated Serving.}
Splitwise~\citep{splitwise} separately provisions prefill and decoding,
DistServe~\citep{distserve} optimizes phase-specific parallelism, and
TetriInfer~\citep{tetriInfer} combines disaggregation with prompt chunking.
Mooncake~\citep{mooncake} and MemServe~\citep{memserve} provide disaggregated
KV cache storage and reuse; SmartGen~\citep{smartgen} studies selective
cache transfer, and HexGen-2~\citep{hexgen2} targets heterogeneous resources.
TaiChi~\citep{taichi}, Arrow~\citep{wu2025arrow}, and
P/D-Serve~\citep{pdserve} adapt serving decisions to workloads and system
conditions. These works do not provide approximation guarantees for
joint admission scheduling.

\noindent\textbf{Theoretical Scheduling Models.}
\citet{jaillet2025online} and \citet{ao2025fluid} study online scheduling
under a unified memory constraint, while \citet{wang2025llm} consider
heterogeneous prefill and decode lengths under a shared cache capacity.
\citet{feng2026competitive} derive competitive and constant-factor
approximation guarantees for non-clairvoyant and clairvoyant scheduling,
respectively, under a single KV cache capacity. Their subsequent work
develops a general non-clairvoyant scheduling
framework~\citep{feng2026regime}.
\citet{chen2025adaptively} consider uncertain output lengths.
\citet{bari2025optimal} study routing and throughput-oriented scheduling,
while \citet{dai2025throughput} analyze the throughput optimality of
work-conserving scheduling. These models do not capture retained prompt
caches across separate prefill and decode memory pools.

\section{Conclusion and Future Work}
\label{sec:conclusion}
We formulated DMRS for disaggregated LLM inference, jointly scheduling prefill
and decode admissions under separate memory capacities while accounting for
prompt KV caches retained in prefill memory and growing decode KV caches. We
proved strong NP-hardness, derived an LP lower bound, and developed two
deterministic polynomial-time approximation algorithms. The first achieves an approximation ratio of
\(8+\frac{4}{1-\alpha}+\frac{8}{1-\beta}\) for
\(\alpha,\beta<1\), while the second provides a uniform
\(36\)-approximation for every feasible instance, including instances
with \(\alpha=1\) or \(\beta=1\).

The CPU experiments show that both algorithms perform well relative to exact ILP solutions and LP lower bounds. Experiments on two H100 GPUs show that both algorithms reduce total
completion time (TCT) relative to vLLM P/D across all tested settings
on code workloads.

\noindent\textbf{Limitations.}
We assume that output token counts are known in advance, consistent
with the theoretical settings studied by \citet{wang2025llm}
and \citet{jaillet2025online}. This assumption allows us to focus on
the coupling between prefill and decode admissions under separate
memory capacities. The theoretical model assumes instantaneous
KV transfer, whereas the completion times measured in our GPU
experiments include KV transfer overhead.

\noindent\textbf{Future Work.}
We can extend DMRS in three directions. An online setting could
address requests arriving over time and establish competitive
guarantees. A non-clairvoyant setting could remove the assumption
that output lengths are known in advance. A deadline-constrained
setting could maximize the total weight of requests completed
by their deadlines.

\section*{AI Use Statement}
Generative AI tools were used to assist with language editing, LaTeX
organization, literature search and reference formatting, and the review and
presentation of the theoretical formulation, algorithms, mathematical claims,
proofs, experimental methodology, and results. The authors independently
reviewed and verified all AI-assisted content, including the cited sources,
derivations, algorithms, proofs, experimental configurations, code, and
reported results. The authors take full responsibility for the final content
of this work.

\section*{Reproducibility Statement}
The model assumptions and algorithms are stated in
Sections~\ref{model_formulation}--\ref{sec:alg2}. Complete proofs of all
theoretical claims, the notation table, and an exact time-indexed ILP
formulation are provided in the appendix. The supplementary CPU and GPU
artifacts provide the fixed configurations, trace processing procedures,
implementation and validation code, frozen results, and scripts used to
reproduce the reported experimental summaries.

\clearpage
\addcontentsline{toc}{section}{References}

\nocite{schrijver2000combinatorial} % Preserve the submitted bibliography on rebuild.
\bibliographystyle{iclr2027_conference}
\bibliography{iclr2027_conference}

\appendix
\clearpage
\setcounter{tocdepth}{2}
\tableofcontents
\clearpage
\addtocontents{toc}{\protect\setcounter{tocdepth}{1}}

\clearpage
\section{Notation}
\label{app:notation}
The notation used throughout the paper is summarized in
Table~\ref{tab:notation}.
\begin{table}[H]
	\centering
	\caption{Summary of notation.}
	\label{tab:notation}
	\small
	\begin{tabular}{@{}lp{0.76\linewidth}@{}}
		\toprule
		Symbol & Definition \\
		\midrule
		\(\mathcal I=\{1,\ldots,n\}\)
		& The set of LLM inference requests \\
		\(m_i\)
		& The effective prefill processing time of request \(i\) \\
		\(p_i\)
		& The number of input tokens (i.e., prefill memory demand) of request \(i\) \\
		\(d_i\)
		& The number of output tokens (i.e., decode iterations) of request \(i\) \\
		\(h_i=p_i+d_i\)
		& The peak decode memory demand of request \(i\) \\
		\(r_i(u)=p_i+u\)
		& The decode memory demand after request \(i\) generates \(u\) output tokens \\
		\(M_P,M_D\)
		& The prefill and decode memory capacities \\
		\(S_i^P,S_i^D\)
		& The prefill and decode admission times of request \(i\) \\
		\(C_i^P,C_i\)
		& The prefill and request completion times of request \(i\) \\
		\(\mathcal U(t)\)
		& The set of requests not yet admitted to prefill during slot \([t,t+1)\) \\
		\(\mathcal E(t)\)
		& The set of requests executing prefill during slot \([t,t+1)\) \\
		\(\mathcal W(t)\)
		& The set of requests waiting for decode admission during slot \([t,t+1)\) \\
		\(\mathcal D(t)\)
		& The set of requests executing decoding during slot \([t,t+1)\) \\
		\(\ell_i^P,\ell_i^D\)
		& The mandatory prefill and decode memory-time loads of request \(i\) \\
		\(\alpha,\beta\)
		& The maximum normalized prefill and peak decode memory demands, respectively \\
		\(\bar C_i\)
		& The relaxed completion-time variable of request \(i\) \\
		\(\mathrm{LP}^{\star},\mathrm{OPT}\)
		& The optimal objective values of \(\mathbf{(LP)}\) and
		\(\mathbf{(O)}\), respectively \\
		\(\mathcal G_q,\mathcal G_{\leq k}\)
		& The dyadic group and prefix in Algorithm~\ref{alg:bounded_scheduling} \\
		\(\mathcal L_P,\mathcal S_P\)
		& The prefill-large and prefill-small request sets \\
		\(\mathcal L_D,\mathcal S_D\)
		& The decode-large and decode-small request sets \\
		\(\mathcal C^{LL},\mathcal C^{SL},\mathcal C^{LS},\mathcal C^{SS}\)
		& The four memory-demand classes \\
		\(\rho,B_q\)
		& The geometric grouping ratio \(\rho\) and boundary \(B_q=\rho^q\) \\
		\(\mathcal H_q,\mathcal H_{\leq k}\)
		& The geometric group and prefix in Algorithm~\ref{alg:general_scheduling} \\
		\bottomrule
	\end{tabular}
\end{table}

\section{Proof of Strong NP-Hardness}
\label{app:proof:thm-hardness}
\begin{proof}
	We reduce from the strongly NP-complete 3-PARTITION
	problem~\citep{garey1979computers}. Given \(3\nu\) positive integers
	\(a_1,\ldots,a_{3\nu}\) satisfying
	\(\sum_{i=1}^{3\nu}a_i=\nu B\) and \(B/4<a_i<B/2\) for every \(i\),
	3-PARTITION asks whether they can be partitioned into \(\nu\) disjoint
	triples, each summing to \(B\).
	
	Construct a DMRS instance with \(\mathcal I=\{1,\ldots,3\nu\}\). For each
	\(a_i\), define a request \(i\) with \(m_i=d_i=1\) and \(p_i=a_i\). Set
	\(M_P=B\), \(M_D=\sum_{i=1}^{3\nu}(a_i+1)=\nu B+3\nu\), and
	\(K=3\nu(\nu+3)/2\). Since
	\(\sum_{i\in\mathcal I}h_i=M_D\), the decode memory constraint never
	prevents decode admission.
	
	Suppose that the 3-PARTITION instance has a valid partition, with its
	triples indexed by \(r=1,\ldots,\nu\). Admit the three requests in triple
	\(r\) to prefill at time \(r-1\) and to decoding at time \(r\). Their total
	prefill memory demand is \(B\). At time \(r\), their decode admission
	releases this prefill memory, allowing the next triple, if any, to be
	admitted to prefill at the same time. The requests in triple \(r\) complete
	at time \(r+1\). Hence, the resulting schedule is feasible and has objective
	value
	\(\sum_{i\in\mathcal I}C_i
	=3\sum_{r=1}^{\nu}(r+1)=3\nu(\nu+3)/2=K\).
	
	Conversely, suppose that the constructed DMRS instance has a feasible
	schedule satisfying \(\sum_{i\in\mathcal I}C_i\leq K\). Let
	\(C_{(1)}\leq\cdots\leq C_{(3\nu)}\) denote the completion times in
	nondecreasing order. Since \(p_i>B/4\), at most three requests can be
	admitted to prefill at any slot boundary. The first \(j\) completed requests
	must have been admitted to prefill by time \(C_{(j)}-2\), because both
	phases require one slot. Their prefill admission times therefore lie in
	\(\{0,\ldots,C_{(j)}-2\}\), which contains \(C_{(j)}-1\) slot boundaries.
	Hence, \(j\leq3(C_{(j)}-1)\), so
	\(C_{(j)}\geq1+\lceil j/3\rceil\). Consequently,
	\[
	\sum_{i\in\mathcal I}C_i
	=
	\sum_{j=1}^{3\nu}C_{(j)}
	\geq
	\sum_{j=1}^{3\nu}
	\left(1+\left\lceil\frac{j}{3}\right\rceil\right)
	=
	K.
	\]
	Thus, equality holds throughout, and in particular
	\(C_{(3\nu)}=\nu+1\). Since \(m_i=d_i=1\), every request satisfies
	\(S_i^P\leq C_i-2\leq C_{(3\nu)}-2=\nu-1\).
	
	For \(t=0,\ldots,\nu-1\), define
	\(\mathcal A_t=\{i\in\mathcal I:S_i^P=t\}\). These sets partition
	\(\mathcal I\), and prefill memory feasibility gives
	\(\sum_{i\in\mathcal A_t}p_i\leq B\). Moreover,
	\[
	\sum_{t=0}^{\nu-1}\sum_{i\in\mathcal A_t}p_i
	=
	\sum_{i\in\mathcal I}p_i
	=
	\nu B.
	\]
	Therefore, \(\sum_{i\in\mathcal A_t}p_i=B\) for every
	\(t=0,\ldots,\nu-1\). Since \(B/4<p_i<B/2\), each \(\mathcal A_t\)
	contains exactly three requests. The sets
	\(\mathcal A_0,\ldots,\mathcal A_{\nu-1}\) therefore yield a valid
	3-partition.
	
	The reduction is computable in polynomial time, and its numerical
	parameters remain polynomially bounded under the strongly NP-complete
	restriction of 3-PARTITION. Therefore, DMRS is strongly NP-hard.
\end{proof}

\section{Proofs for the LP Relaxation}
\label{app:lp-proofs}
\subsection{Proof of Lemma~\ref{lem:lp_lower_bound}}
\label{app:proof:lem-lp_lower_bound}
\begin{proof}
Constraint~(\plaineqref{eq:lp_single}) follows from
\(C_i=S_i^D+d_i\geq S_i^P+m_i+d_i\geq m_i+d_i\), where the two
inequalities follow from \(S_i^D\geq S_i^P+m_i\) and \(S_i^P\geq0\),
respectively.

Consider a nonempty subset \(\mathcal S\subseteq\mathcal I\) and either
memory pool. Let \(M\) denote the capacity of the selected pool, and let
\(\ell_i\) denote the corresponding mandatory memory-time load. Reindex the
	requests in \(\mathcal S\) as \(1,\ldots,s\), where
	\(s=|\mathcal S|\), such that \(C_1\leq\cdots\leq C_s\). For each
	\(j=1,\ldots,s\), define \(L_j=\sum_{r=1}^{j}\ell_r\).
	
	For the prefill pool, the mandatory load of request \(r\) is completed by
	\(C_r^P\leq C_r\). For the decode pool, its mandatory load is completed by
	\(C_r\). Therefore, by time \(C_j\), the selected pool must have accommodated
	the total load \(L_j\) of requests \(1,\ldots,j\). Since the pool has capacity
	\(M\) in each slot, it can accommodate at most \(MC_j\) units of memory-time
	load by that time. Hence, \(MC_j\geq L_j\).
	
	Dividing this inequality by \(M\), multiplying it by \(\ell_j\), and summing
	over \(j=1,\ldots,s\) gives
	\[
	\sum_{j=1}^{s}\ell_j C_j
	\geq
	\frac{1}{M}\sum_{j=1}^{s}\ell_j L_j
	=
	\frac{1}{2M}
	\left[
	\left(\sum_{j=1}^{s}\ell_j\right)^2
	+
	\sum_{j=1}^{s}\ell_j^2
	\right].
	\]
	The equality follows from
	\(L_j=\sum_{r=1}^{j}\ell_r\).
	
	Taking \((\ell_i,M)=(\ell_i^P,M_P)\) proves
	Constraint~(\plaineqref{eq:lp_prefill_subset}), while taking
	\((\ell_i,M)=(\ell_i^D,M_D)\) proves
	Constraint~(\plaineqref{eq:lp_decode_subset}). Thus, every feasible schedule of
	\(\mathbf{(O)}\) induces a feasible solution to \(\mathbf{(LP)}\) with
	\(\bar C_i=C_i\) for all \(i\in\mathcal I\). Therefore,
	\(\mathrm{LP}^{\star}\leq\mathrm{OPT}\).
\end{proof}

\subsection{Proof of Lemma~\ref{lem:lp_polynomial}}
\label{app:proof:lem-lp_polynomial}
\begin{proof}
	Although \(\mathbf{(LP)}\) contains exponentially many subset constraints,
	they admit a polynomial-time separation oracle. Given a candidate vector
	\(\bar{\mathbf C}=(\bar C_i)_{i\in\mathcal I}\), consider either memory pool.
	Let \(M\) denote its capacity and \(\ell_i\) the corresponding mandatory
	memory-time load. For each \(\mathcal S\subseteq\mathcal I\), define
	\[
	F(\mathcal S)
	:=
	2M\sum_{i\in\mathcal S}\ell_i\bar C_i
	-
	\left(\sum_{i\in\mathcal S}\ell_i\right)^2
	-
	\sum_{i\in\mathcal S}\ell_i^2.
	\]
	For every nonempty \(\mathcal S\), the corresponding subset constraint is
	violated if and only if \(F(\mathcal S)<0\).
	
	An exact separation oracle follows by sorting. Set
	\(a_i=M\bar C_i-\ell_i/2\) and
	\(L(\mathcal S)=\sum_{i\in\mathcal S}\ell_i\). Then
	\(F(\mathcal S)/2=\sum_{i\in\mathcal S}\ell_i a_i
	-L(\mathcal S)^2/2\). Let \(\mathcal S^\star\) minimize \(F\) over all
	subsets and write \(L=L(\mathcal S^\star)\). Removing any
	\(i\in\mathcal S^\star\) cannot improve the minimum, so
	\(a_i\leq L-\ell_i/2\). Likewise, adding any
	\(j\notin\mathcal S^\star\) cannot improve it, so
	\(a_j\geq L+\ell_j/2\). Since all loads are positive, every member of
	\(\mathcal S^\star\) precedes every nonmember when requests are sorted by
	\(a_i\). Thus, some minimizing subset is a prefix of this order. Sorting
	and checking all prefixes, including the empty prefix with value zero,
	takes \(O(n\log n)\) time per pool. A negative minimum identifies a
	violated subset constraint; a zero minimum certifies all subset
	constraints for that pool. Checking both pools and
	Constraint~(\plaineqref{eq:lp_single}) gives a polynomial-time
	separation oracle.
	
	It remains to bound the search region. Let
	\(T_0:=\sum_{i\in\mathcal I}(m_i+d_i)\). Consider a schedule that processes
	the requests one at a time and completes both phases of each request before
	starting the next. By the feasibility
	conditions stated in Section~\ref{sec:system_model}, this schedule is
	feasible, and every request completes by time \(T_0\). By
	Lemma~\ref{lem:lp_lower_bound}, the corresponding completion-time vector is
	feasible for \(\mathbf{(LP)}\) and has objective value at most \(nT_0\).
	Therefore, \(\mathrm{LP}^{\star}\leq nT_0\). Since
	\(\bar C_i\geq m_i+d_i>0\), every optimal solution satisfies
	\[
	0\leq\bar C_i\leq nT_0,
	\qquad i\in\mathcal I.
	\]
	These bounds have encoding length polynomial in the input size and may be
	added to \(\mathbf{(LP)}\) without changing its optimal value. They can also
	be separated directly. The equivalence of optimization and separation
	therefore implies that \(\mathbf{(LP)}\) can be solved in polynomial time
	using the ellipsoid method~\citep{groetschel1993geometric}.
\end{proof}

\paragraph{LP Implementation.}
In the CPU experiments, we solve \(\mathbf{(LP)}\) using Gurobi and
an iterative cutting plane procedure. We first solve the LP without
the subset constraints. We then apply the exact prefix separator to
both memory pools. At each iteration, we add up to 100 previously
unseen prefix constraints with the largest violations for each pool.
We solve the updated LP and repeat this process until no checked
prefix constraint has a relative violation exceeding \(10^{-8}\).

For the GPU experiments, we use HiGHS to solve \(\mathbf{(LP)}\). The initial LP
includes the subset constraints for each individual request and for
the full request set in both memory pools. We then add violated
prefix constraints iteratively.

\section{Detailed Algorithm Descriptions}
\label{app:algorithm-descriptions}

\subsection{Description of Algorithm~\ref{alg:bounded_scheduling}}
\label{app:alg1-pseudocode}
Line~1 solves the LP relaxation, obtains an optimal solution
\(\{\bar C_i\}_{i\in\mathcal I}\), and constructs the nonempty dyadic groups.
Line~2 initializes the current time and the active request sets.

Lines~3--24 process the nonempty groups in increasing order of \(q\). For
each group \(\mathcal G_q\), Line~4 initializes
\(\mathcal U_q=\mathcal G_q\). Lines~5--23 process the current group until
\(\mathcal U_q\), \(\mathcal E\), \(\mathcal W\), and \(\mathcal D\) are all
empty. At each event time \(t\), Line~6 removes the requests that have
completed decoding from \(\mathcal D\), and Line~7 moves the requests that
have completed prefill from \(\mathcal E\) to \(\mathcal W\).

Lines~8--13 implement the decode admission rule. Line~8 creates a copy
\(\widehat{\mathcal W}\) of the current waiting set. Lines~9--13 consider
the requests in \(\widehat{\mathcal W}\) in nondecreasing order of
\(m_i+d_i\), with ties broken by request index. Line~10 checks the peak
reservation condition: request \(i\) is admitted only if its peak demand \(h_i\),
together with the peak demands \(h_j\) reserved by all active requests
\(j\in\mathcal D\), does not exceed \(M_D\). Because each request remains in
\(\mathcal D\) until its decode completion event at Line~6, its peak demand
remains reserved throughout decoding. Upon admission,
Line~11 sets \(S_i^D=t\), removes \(i\) from \(\mathcal W\), and adds it to
\(\mathcal D\).

Lines~14--19 implement the prefill admission rule. Line~14 creates a copy
\(\widehat{\mathcal U}_q\) of the current unstarted set. Lines~15--19 consider
the requests in \(\widehat{\mathcal U}_q\) in nondecreasing order of
\(m_i+d_i\), with ties broken by request index. A request \(i\) is admitted to
prefill if its demand \(p_i\), together with the prefill memory occupied by
the requests in \(\mathcal E\) and \(\mathcal W\), does not exceed \(M_P\).
Upon admission, Line~17 sets \(S_i^P=t\), removes \(i\) from
\(\mathcal U_q\), and adds it to \(\mathcal E\).

If the current group remains unfinished after both admission rules,
Lines~20--22 advance \(t\) to the next prefill or decode completion time.

\subsection{Description of Algorithm~\ref{alg:general_scheduling}}
\label{app:alg2-pseudocode}
\label{app:alg2-description}

Line~1 solves \(\mathbf{(LP)}\), obtains an optimal solution, and constructs
the four request classes. Line~2 forms the nonempty geometric groups, and
Line~3 sets \(t=0\) and initializes
\(\mathcal E\), \(\mathcal W\), and \(\mathcal D\) as empty sets.

Line~4 processes the nonempty groups in increasing order of \(q\).
Lines~5--7 process the requests in
\(\mathcal H_q\cap\mathcal C^{LL}\) one at a time, in nondecreasing order of
\(m_i+d_i\). For each request, Line~6 sets its admission times and advances
\(t\) to its completion time.

Lines~8--29 process the remaining classes in the order
\(\mathcal C^{SL}\), \(\mathcal C^{LS}\), and \(\mathcal C^{SS}\).
Line~9 collects the unstarted requests in the current group and class, and
the loop in Lines~10--28 continues until all of them have completed. At each
event time, Lines~11--12 remove completed decode requests and move completed
prefill requests to \(\mathcal W\).

Lines~13--18 implement the decode admission rule. Line~13 creates a copy
\(\widehat{\mathcal W}\) of the current waiting set. Lines~14--18 consider
the requests in \(\widehat{\mathcal W}\) in nondecreasing order of
\(m_i+d_i\), with ties broken by request index. Line~15 checks the peak
reservation condition: request \(i\) is admitted only if
\(\sum_{j\in\mathcal D}h_j+h_i\leq M_D\). Because each request remains in
\(\mathcal D\) until its decode completion event at Line~11, its peak demand
remains reserved throughout decoding. Upon admission, Line~16 sets its decode
admission time and moves it from \(\mathcal W\) to \(\mathcal D\).

Lines~19--24 implement the prefill admission rule. Line~19 creates a copy
\(\widehat{\mathcal U}_{q,\mathcal C}\) of the unstarted requests in the
current group and class. Lines~20--24 consider these requests in
nondecreasing order of \(m_i+d_i\), with ties broken by request index. A
request is admitted when its prefill memory demand can be accommodated together
with the memory occupied by the requests in \(\mathcal E\) and \(\mathcal W\).
Upon admission, Line~22 sets its prefill admission time and moves it from
\(\mathcal U_{q,\mathcal C}\) to \(\mathcal E\).

If the current class remains unfinished, Lines~25--27 advance \(t\) to the
next prefill or decode completion time.

\section{Proofs for Algorithm~\ref{alg:bounded_scheduling}}
\label{app:alg1-proofs}
We first bound the cumulative memory-time loads of the dyadic prefixes
defined in the proof of Theorem~\ref{thm:alg1_ratio}.

\begin{lemma}
	For each \(k\in\mathbb Z\), the total prefill and decode memory-time loads
	of the prefix \(\mathcal G_{\leq k}\) satisfy
	\(\sum_{i\in\mathcal G_{\leq k}}\ell_i^P\le 2M_P\,2^k\) and
	\(\sum_{i\in\mathcal G_{\leq k}}\ell_i^D\le 2M_D\,2^k\), respectively.
	\label{lem:alg1_prefix_load}
\end{lemma}

\begin{proof}
	If \(\mathcal G_{\leq k}=\emptyset\), the claim is immediate. Otherwise,
	each \(i\in\mathcal G_{\leq k}\) belongs to some
	\(\mathcal G_q\) with \(q\le k\), and hence
	\(\bar C_i\le 2^q\le 2^k\). Therefore,
	\[
	\sum_{i\in\mathcal G_{\leq k}}
	\ell_i^P\bar C_i
	\le
	2^k
	\sum_{i\in\mathcal G_{\leq k}}\ell_i^P.
	\]
	Applying Constraint~(\plaineqref{eq:lp_prefill_subset}) to
	\(\mathcal G_{\leq k}\) and dropping the nonnegative term
	\(\sum_{i\in\mathcal G_{\leq k}}(\ell_i^P)^2\) gives
	\[
	\sum_{i\in\mathcal G_{\leq k}}
	\ell_i^P\bar C_i
	\ge
	\frac{1}{2M_P}
	\left(
	\sum_{i\in\mathcal G_{\leq k}}\ell_i^P
	\right)^2.
	\]
	Combining the two
	inequalities yields
	\[
	\sum_{i\in\mathcal G_{\leq k}}\ell_i^P
	\le
	2M_P\,2^k.
	\]
	Applying the same argument to
	Constraint~(\plaineqref{eq:lp_decode_subset}) gives
	\[
	\sum_{i\in\mathcal G_{\leq k}}\ell_i^D
	\le
	2M_D\,2^k,
	\]
	which completes the proof.
\end{proof}
	
Algorithm~\ref{alg:bounded_scheduling} reserves the peak decode memory demand
\(h_i\) throughout the \(d_i\) decode slots of request \(i\).
Lemma~\ref{lem:alg1_decode_reservation} shows that this reserved memory-time
load is less than twice the mandatory decode memory-time load captured by the
LP.

\begin{lemma}
	For each request \(i\in\mathcal I\), its reserved decode memory-time
	\(h_i d_i\) satisfies
	\(\ell_i^D\le h_i d_i<2\ell_i^D\). Consequently, for every
	\(k\in\mathbb Z\),
	\(\sum_{i\in\mathcal G_{\leq k}}h_i d_i<4M_D\,2^k\).
	\label{lem:alg1_decode_reservation}
\end{lemma}

\begin{proof}
	Using \(h_i=p_i+d_i\) and
	\(\ell_i^D=p_i d_i+d_i(d_i+1)/2\), we have
	\[
	h_i d_i-\ell_i^D
	=
	\frac{d_i(d_i-1)}{2}
	\ge 0,
	\qquad
	2\ell_i^D-h_i d_i
	=
	d_i(p_i+1)
	>0.
	\]
	Hence, \(\ell_i^D\le h_i d_i<2\ell_i^D\).
	If \(\mathcal G_{\leq k}=\emptyset\), the prefix bound is immediate.
	Otherwise, summing over \(\mathcal G_{\leq k}\) and applying
	Lemma~\ref{lem:alg1_prefix_load} gives
	\[
	\sum_{i\in\mathcal G_{\leq k}}h_i d_i
	<
	2\sum_{i\in\mathcal G_{\leq k}}\ell_i^D
	\le
	4M_D\,2^k,
	\]
	which completes the proof.
\end{proof}
	
\begin{lemma}
	For each \(k\in\mathbb Z\), Algorithm~\ref{alg:bounded_scheduling} completes
	all requests in \(\mathcal G_{\leq k}\) by time
	\[
	\left(
	4+\frac{2}{1-\alpha}+\frac{4}{1-\beta}
	\right)2^k.
	\]
	\label{lem:alg1_prefix_completion}
\end{lemma}

\begin{proof}
	If \(\mathcal G_{\leq k}=\emptyset\), the claim is immediate. Otherwise,
	let \(T_k\) denote the time at which all requests in
	\(\mathcal G_{\leq k}\) have completed. By the strict group barrier, only
	requests in \(\mathcal G_{\leq k}\) are processed during \([0,T_k)\).
	
	For each slot \(t<T_k\), let \(\mathcal U_q(t)\), \(\mathcal E(t)\),
	\(\mathcal W(t)\), and \(\mathcal D(t)\) denote the request sets during
	\([t,t+1)\), where \(\mathcal G_q\) is the group currently being processed.
	Define
	\[
	\begin{aligned}
		\mathcal T_1
		&:=
		\{t<T_k:\mathcal W(t)\neq\emptyset\},\\
		\mathcal T_2
		&:=
		\{t<T_k:\mathcal W(t)=\emptyset,\,
		\mathcal U_q(t)\neq\emptyset\},\\
		\mathcal T_3
		&:=
		\{t<T_k:\mathcal W(t)=\mathcal U_q(t)=\emptyset,\,
		\mathcal E(t)\neq\emptyset\},\\
		\mathcal T_4
		&:=
		\{t<T_k:\mathcal W(t)=\mathcal U_q(t)=\mathcal E(t)=\emptyset,\,
		\mathcal D(t)\neq\emptyset\}.
	\end{aligned}
	\]
	For every \(t\in\{0,\ldots,T_k-1\}\), at least one of
	\(\mathcal U_q(t)\), \(\mathcal E(t)\), \(\mathcal W(t)\), and
	\(\mathcal D(t)\) is nonempty. Therefore,
	\(\mathcal T_1,\ldots,\mathcal T_4\) are pairwise disjoint and partition the
	slots in \([0,T_k)\). Here, \(|\mathcal T_r|\) denotes the number of slots
	in \(\mathcal T_r\), for \(r=1,\ldots,4\).
	
	\emph{Case 1: \(t\in\mathcal T_1\).}
	For any \(t\in\mathcal T_1\), choose a request
	\(i\in\mathcal W(t)\) that remains waiting for decode admission. By the
	decode admission rule in Algorithm~\ref{alg:bounded_scheduling},
	\(\sum_{j\in\mathcal D(t)}h_j+h_i>M_D\).
	Since \(h_i\leq\beta M_D\), it follows that
	\[
	\sum_{j\in\mathcal D(t)}h_j>(1-\beta)M_D.
	\]
	Summing this inequality over \(t\in\mathcal T_1\) gives
	\[
	|\mathcal T_1|(1-\beta)M_D
	\leq
	\sum_{t\in\mathcal T_1}\sum_{j\in\mathcal D(t)}h_j
	\leq
	\sum_{i\in\mathcal G_{\leq k}}h_i d_i
	<
	4M_D\,2^k,
	\]
	where the final inequality follows from
	Lemma~\ref{lem:alg1_decode_reservation}. Therefore,
	\[
	|\mathcal T_1|<\frac{4}{1-\beta}2^k.
	\]
	
	\emph{Case 2: \(t\in\mathcal T_2\).}
	For any \(t\in\mathcal T_2\), choose a request
	\(i\in\mathcal U_q(t)\) that remains unadmitted to prefill. Because
	\(\mathcal W(t)=\emptyset\), the prefill admission rule in Algorithm~\ref{alg:bounded_scheduling} implies
	\(\sum_{j\in\mathcal E(t)}p_j+p_i>M_P\).
	Since \(p_i\leq\alpha M_P\), it follows that
	\[
	\sum_{j\in\mathcal E(t)}p_j>(1-\alpha)M_P.
	\]
	Summing this inequality over \(t\in\mathcal T_2\) gives
	\[
	|\mathcal T_2|(1-\alpha)M_P
	\leq
	\sum_{t\in\mathcal T_2}\sum_{j\in\mathcal E(t)}p_j
	\leq
	\sum_{i\in\mathcal G_{\leq k}}\ell_i^P
	\leq
	2M_P\,2^k,
	\]
	where the final inequality follows from
	Lemma~\ref{lem:alg1_prefix_load}. Therefore,
	\[
	|\mathcal T_2|\leq\frac{2}{1-\alpha}2^k.
	\]
	
	\emph{Case 3: \(t\in\mathcal T_3\).}
	Fix a nonempty group \(\mathcal G_q\). Suppose that
	\(\mathcal T_3\) contains a slot during which \(\mathcal G_q\) is being
	processed, and consider the first such slot. At this time,
	\(\mathcal U_q(t)=\emptyset\), so no additional request from
	\(\mathcal G_q\) can enter prefill. The group barrier also prevents requests
	from later groups from entering prefill. For every
	\(i\in\mathcal G_q\), Constraint~(\plaineqref{eq:lp_single}) and the
	definition of \(\mathcal G_q\) imply
	\[
	m_i\leq m_i+d_i\leq\bar C_i\leq2^q.
	\]
	Thus, every request executing prefill in this slot completes prefill within
	at most \(2^q\) additional slots. Since no new request can enter prefill,
	the number of slots in \(\mathcal T_3\) during which
	\(\mathcal G_q\) is processed is at most \(2^q\).
	
	\emph{Case 4: \(t\in\mathcal T_4\).}
	Fix a nonempty group \(\mathcal G_q\). Suppose that
	\(\mathcal T_4\) contains a slot during which \(\mathcal G_q\) is being
	processed, and consider the first such slot. At this time,
	\(\mathcal U_q(t)=\mathcal W(t)=\mathcal E(t)=\emptyset\), so only requests
	executing decoding remain. No additional request can enter decoding
	thereafter. For every \(i\in\mathcal G_q\),
	Constraint~(\plaineqref{eq:lp_single}) and the definition of
	\(\mathcal G_q\) imply
	\[
	d_i\leq m_i+d_i\leq\bar C_i\leq2^q.
	\]
	Thus, every remaining request completes decoding within at most \(2^q\)
	additional slots. Therefore, the number of slots in \(\mathcal T_4\) during
	which \(\mathcal G_q\) is processed is at most \(2^q\).
	
	Since
	\[
	\sum_{\substack{q\leq k\\\mathcal G_q\neq\emptyset}}2^q
	\leq
	\sum_{q=-\infty}^{k}2^q
	=
	2^{k+1},
	\]
	combining the four cases gives
	\[
	\begin{aligned}
		T_k
		&=
		|\mathcal T_1|+|\mathcal T_2|
		+|\mathcal T_3|+|\mathcal T_4|<
		\frac{4}{1-\beta}2^k
		+
		\frac{2}{1-\alpha}2^k
		+
		2\sum_{\substack{q\leq k\\\mathcal G_q\neq\emptyset}}2^q\\
		&\leq
		\frac{4}{1-\beta}2^k
		+
		\frac{2}{1-\alpha}2^k
		+
		4\cdot2^k=
		\left(
		4+\frac{2}{1-\alpha}+\frac{4}{1-\beta}
		\right)2^k.
	\end{aligned}
	\]
	This completes the proof.
\end{proof}

\section{Proofs for Algorithm~\ref{alg:general_scheduling}}
\label{app:alg2-proofs}
\paragraph{Slot Categories.}
To bound the completion time of a geometric prefix, we assign every execution
slot, after the admission decisions at its starting time, to one charging
category. During the \(\mathcal C^{LL}\) phase of group \(\mathcal H_q\), each
prefill slot is assigned to the large prefill category and each decode slot
to the large decode category. For each remaining class
\(\mathcal C\in\{\mathcal C^{SL},\mathcal C^{LS},\mathcal C^{SS}\}\), consider
the phase that processes \(\mathcal H_q\cap\mathcal C\). Let
\(\mathcal U_{q,\mathcal C}(t)\), \(\mathcal E(t)\), \(\mathcal W(t)\), and
\(\mathcal D(t)\) denote the request sets after admission at time \(t\). The
slots in these three class phases are assigned by the following priority
rules.

\begin{itemize}
	\item
	For \(\mathcal C^{SL}\), assign the slot to the large decode category if
	\(\mathcal D(t)\ne\emptyset\). Otherwise, assign it to the blocked small
	prefill category if \(\mathcal U_{q,\mathcal C}(t)\ne\emptyset\), and to the
	prefill tail category in the remaining case.
	
	\item
	For \(\mathcal C^{LS}\), assign the slot to the blocked small decode
	category if \(\mathcal W(t)\ne\emptyset\). Otherwise, assign it to the
	large prefill category if \(\mathcal E(t)\ne\emptyset\), and to the
	decode tail category in the remaining case.
	
	\item
	For \(\mathcal C^{SS}\), assign the slot to the blocked small decode
	category if \(\mathcal W(t)\ne\emptyset\). Otherwise, assign it to the
	blocked small prefill category if
	\(\mathcal U_{q,\mathcal C}(t)\ne\emptyset\), to the
	prefill tail category if \(\mathcal E(t)\ne\emptyset\), and to the
	decode tail category in the remaining case.
\end{itemize}

Under this assignment, a prefill tail slot satisfies
\(\mathcal U_{q,\mathcal C}(t)=\mathcal W(t)=\emptyset\) and
\(\mathcal E(t)\ne\emptyset\); requests in \(\mathcal D(t)\) may execute
concurrently. A decode tail slot satisfies
\(\mathcal U_{q,\mathcal C}(t)=\mathcal E(t)=\mathcal W(t)=\emptyset\) and
\(\mathcal D(t)\ne\emptyset\).
These rules are exhaustive. In a \(\mathcal C^{SL}\) phase, if
\(\mathcal D(t)=\emptyset\), then \(\mathcal W(t)=\emptyset\), since any
waiting request fits in an empty decode pool. If
\(\mathcal U_{q,\mathcal C}(t)=\emptyset\) as well, every unfinished request
must be in \(\mathcal E(t)\), which gives a prefill tail. In a
\(\mathcal C^{LS}\) phase, if \(\mathcal W(t)=\mathcal E(t)=\emptyset\), then
\(\mathcal U_{q,\mathcal C}(t)=\emptyset\), since any unstarted request fits
in an empty prefill pool. Every unfinished request must therefore be in
\(\mathcal D(t)\), which gives a decode tail. The four priority cases for
\(\mathcal C^{SS}\) directly cover every possible unfinished state. The
strict group and class barriers ensure that no request outside the current
phase affects this classification.

\begin{lemma}
	\label{lem:alg2_prefix_load}
	For every \(k\in\mathbb Z\), the prefix \(\mathcal H_{\leq k}\) satisfies
	\(\sum_{i\in\mathcal H_{\leq k}}\ell_i^P
	\leq 2M_P B_k\) and \(
	\sum_{i\in\mathcal H_{\leq k}}\ell_i^D
	\leq
	2M_D B_k.\)
\end{lemma}

\begin{proof}
	The claim is immediate when
	\(\mathcal H_{\leq k}=\emptyset\). Otherwise, every
	\(i\in\mathcal H_{\leq k}\) satisfies \(\bar C_i\leq B_k\).
	Applying Constraint~(\plaineqref{eq:lp_prefill_subset}) to
	\(\mathcal H_{\leq k}\) gives
	\[
	\frac{1}{2M_P}
	\left(
	\sum_{i\in\mathcal H_{\leq k}}\ell_i^P
	\right)^2
	\leq
	\sum_{i\in\mathcal H_{\leq k}}\ell_i^P\bar C_i
	\leq
	B_k
	\sum_{i\in\mathcal H_{\leq k}}\ell_i^P.
	\]
	Since the total prefill load is positive, division gives
	\[
	\sum_{i\in\mathcal H_{\leq k}}\ell_i^P
	\leq
	2M_P B_k.
	\]
	Replacing \((\ell_i^P,M_P)\) with \((\ell_i^D,M_D)\) and applying
	Constraint~(\plaineqref{eq:lp_decode_subset}) proves the decode bound.
\end{proof}

\begin{lemma}
	\label{lem:alg2_prefix_completion}
	For every \(k\in\mathbb Z\),
	Algorithm~\ref{alg:general_scheduling} completes all requests in
	\(\mathcal H_{\leq k}\) by time
	\(\left(
	12+\frac{4\rho}{\rho-1}
	\right)B_k.\)
\end{lemma}

\begin{proof}
	The claim is immediate when
	\(\mathcal H_{\leq k}=\emptyset\). Otherwise, define
	\(T_k:=\max_{i\in\mathcal H_{\leq k}}C_i\), the time at which the algorithm
	completes the prefix \(\mathcal H_{\leq k}\). The strict group barrier
	implies that only requests in this prefix are processed during
	\([0,T_k)\).
	
	Let \(\mathcal T_{PL}\), \(\mathcal T_{PS}\),
	\(\mathcal T_{DL}\), and \(\mathcal T_{DS}\) denote the sets of slots
	assigned to the large prefill, blocked small prefill, large decode, and
	blocked small decode categories, respectively. Let
	\(\mathcal T_{\mathrm{tail}}\) contain all prefill and decode tail slots.
	
	\emph{1. Prefill Categories.}
	For every \(i\in\mathcal L_P\), \(p_i>M_P/2\), and hence
	\[
	m_i
	=
	\frac{\ell_i^P}{p_i}
	<
	\frac{2\ell_i^P}{M_P}.
	\]
	Therefore,
	\[
	|\mathcal T_{PL}|
	\leq
	\sum_{i\in\mathcal H_{\leq k}\cap\mathcal L_P}m_i
	\leq
	\frac{2}{M_P}
	\sum_{i\in\mathcal H_{\leq k}\cap\mathcal L_P}\ell_i^P.
	\]
	
	For every \(t\in\mathcal T_{PS}\), a prefill-small request remains
	unadmitted while \(\mathcal W(t)=\emptyset\). If \(i\) is such a request,
	then \(p_i\leq M_P/2\), and the failed prefill admission implies
	\(\sum_{j\in\mathcal E(t)}p_j+p_i>M_P\).
	Consequently,
	\(\sum_{j\in\mathcal E(t)}p_j>\frac{M_P}{2}\).
	All requests in \(\mathcal E(t)\) belong to \(\mathcal S_P\). Summing over
	\(\mathcal T_{PS}\) gives
	\[
	|\mathcal T_{PS}|\frac{M_P}{2}
	\leq
	\sum_{t\in\mathcal T_{PS}}
	\sum_{j\in\mathcal E(t)}p_j
	\leq
	\sum_{i\in\mathcal H_{\leq k}\cap\mathcal S_P}\ell_i^P.
	\]
	Combining the two prefill bounds and applying
	Lemma~\ref{lem:alg2_prefix_load} yields
	\[
	|\mathcal T_{PL}|+|\mathcal T_{PS}|
	\leq
	\frac{2}{M_P}
	\sum_{i\in\mathcal H_{\leq k}}\ell_i^P
	\leq
	4B_k.
	\]
	
	\emph{2. Decode Categories.}
	Under peak decode reservation, each admitted request \(i\) reserves \(h_i\)
	until its completion. Hence, no two decode-large requests can overlap in the
	decode pool.
	For every \(i\in\mathcal L_D\), \(h_i>M_D/2\). By
	Lemma~\ref{lem:alg1_decode_reservation},
	\(h_i d_i<2\ell_i^D\), and hence
	\(d_i
	<
	\frac{2h_i d_i}{M_D}
	<
	\frac{4\ell_i^D}{M_D}\).
	Therefore,
	\[
	|\mathcal T_{DL}|
	\leq
	\sum_{i\in\mathcal H_{\leq k}\cap\mathcal L_D}d_i
	\leq
	\frac{4}{M_D}
	\sum_{i\in\mathcal H_{\leq k}\cap\mathcal L_D}\ell_i^D.
	\]
	
	For every \(t\in\mathcal T_{DS}\), a decode-small request remains waiting
	after decode admission. If \(i\) is such a request, then
	\(h_i\leq M_D/2\), and the failed decode admission implies
	\[
	\sum_{j\in\mathcal D(t)}h_j+h_i>M_D.
	\]
	Consequently,
	\(\sum_{j\in\mathcal D(t)}h_j>\frac{M_D}{2}\). All requests in \(\mathcal D(t)\) belong to \(\mathcal S_D\). Summing over
	\(\mathcal T_{DS}\) and using \(h_i d_i<2\ell_i^D\) gives
	\[
	|\mathcal T_{DS}|\frac{M_D}{2}
	\leq
	\sum_{t\in\mathcal T_{DS}}
	\sum_{j\in\mathcal D(t)}h_j
	\leq
	\sum_{i\in\mathcal H_{\leq k}\cap\mathcal S_D}h_i d_i
	\leq
	2\sum_{i\in\mathcal H_{\leq k}\cap\mathcal S_D}\ell_i^D.
	\]
	Combining the two decode bounds and applying
	Lemma~\ref{lem:alg2_prefix_load} yields
	\[
	|\mathcal T_{DL}|+|\mathcal T_{DS}|
	\leq
	\frac{4}{M_D}
	\sum_{i\in\mathcal H_{\leq k}}\ell_i^D
	\leq
	8B_k.
	\]
	
	\emph{3. Tail Categories.}
	For each nonempty group \(\mathcal H_q\), tail slots can occur in four
	class phases: the prefill tail of \(\mathcal C^{SL}\), the decode tail of
	\(\mathcal C^{LS}\), and the prefill and decode tails of
	\(\mathcal C^{SS}\). Every \(i\in\mathcal H_q\) satisfies
	\(m_i+d_i\leq\bar C_i\leq B_q\), and hence \(m_i,d_i\leq B_q\).
	
	When a prefill tail first occurs in a class phase, no unstarted request
	remains, so no additional request can enter prefill during that phase.
	All requests currently executing prefill therefore complete prefill within
	at most \(B_q\) slots. When a decode tail first occurs, no request remains
	in \(\mathcal U_{q,\mathcal C}(t)\), \(\mathcal E(t)\), or
	\(\mathcal W(t)\). Hence, no additional request can enter decoding, and all
	requests currently executing
	decoding complete within at most \(B_q\) slots. Therefore,
	\[
	|\mathcal T_{\mathrm{tail}}|
	\leq
	4\sum_{\substack{q\leq k\\\mathcal H_q\ne\emptyset}}B_q
	\leq
	4\sum_{q\leq k}\rho^q
	=
	\frac{4\rho}{\rho-1}B_k.
	\]
	
	Because the slot categories partition \([0,T_k)\),
	\[
	\begin{aligned}
		T_k
		&=
		|\mathcal T_{PL}|+|\mathcal T_{PS}|
		+|\mathcal T_{DL}|+|\mathcal T_{DS}|
		+|\mathcal T_{\mathrm{tail}}|\\
		&\leq
		4B_k+8B_k+\frac{4\rho}{\rho-1}B_k=
		\left(
		12+\frac{4\rho}{\rho-1}
		\right)B_k.
	\end{aligned}
	\]
	This completes the proof.
\end{proof}

\section{Exact ILP Formulation}
\label{app:ilp_formulation}
The following time-indexed ILP computes exact optimal solutions for every
instance satisfying \(p_i\leq M_P\) and
\(h_i\leq M_D\) for all \(i\in\mathcal I\). Under these conditions, processing
the requests sequentially, with each request entering decoding immediately
after completing prefill, gives a feasible schedule with makespan
\[
H:=\sum_{i\in\mathcal I}(m_i+d_i).
\]

An optimal schedule can be chosen with makespan at most \(H\). Indeed, for
any feasible schedule with makespan
\(C_{\max}:=\max_{i\in\mathcal I}C_i\),
\[
\sum_{t=0}^{C_{\max}-1}
\bigl(|\mathcal E(t)|+|\mathcal D(t)|\bigr)
=
\sum_{i\in\mathcal I}(m_i+d_i)
=
H.
\]
If \(C_{\max}>H\), some slot contains no request executing prefill or
decoding. Deleting this slot and shifting all subsequent events one slot
earlier preserves precedence, does not increase memory usage in either pool,
and strictly decreases at least one completion time, contradicting
optimality.

Define the slot set \(\mathbb T:=\{0,\ldots,H-1\}\) and the feasible prefill
and decode admission-time sets
\(\mathbb T_i^P:=\{0,\ldots,H-m_i-d_i\}\) and
\(\mathbb T_i^D:=\{0,\ldots,H-d_i\}\), respectively.
For each \(i\in\mathcal I\), let \(x_{i,s}^P\) and \(x_{i,s}^D\) indicate
whether request \(i\) is admitted to prefill and decoding at time \(s\),
respectively. The ILP is
\begin{subequations}
	\label{eq:ilp}
	\begin{align}
		\mathbf{(ILP)}
		\qquad
		\min\quad
		&
		\sum_{i\in\mathcal I}
		\sum_{s\in\mathbb T_i^D}
		(s+d_i)x_{i,s}^D
		\label{eq:ilp_objective}\\
		\text{s.t.}\quad
		&
		\sum_{s\in\mathbb T_i^P}x_{i,s}^P=1,
		&&
		\forall i\in\mathcal I,
		\label{eq:ilp_unique_prefill}\\
		&
		\sum_{s\in\mathbb T_i^D}x_{i,s}^D=1,
		&&
		\forall i\in\mathcal I,
		\label{eq:ilp_unique_decode}\\
		&
		\sum_{s\in\mathbb T_i^D}s\,x_{i,s}^D
		\geq
		\sum_{s\in\mathbb T_i^P}(s+m_i)x_{i,s}^P,
		&&
		\forall i\in\mathcal I,
		\label{eq:ilp_precedence}\\
		&
		\sum_{i\in\mathcal I}p_i
		\left(
		\sum_{\substack{s\in\mathbb T_i^P\\s\leq t}}x_{i,s}^P
		-
		\sum_{\substack{s\in\mathbb T_i^D\\s\leq t}}x_{i,s}^D
		\right)
		\leq M_P,
		&&
		\forall t\in\mathbb T,
		\label{eq:ilp_prefill_capacity}\\
		&
		\sum_{i\in\mathcal I}
		\sum_{\substack{s\in\mathbb T_i^D\\s\leq t<s+d_i}}
		(p_i+t-s+1)x_{i,s}^D
		\leq M_D,
		&&
		\forall t\in\mathbb T,
		\label{eq:ilp_decode_capacity}\\
		&
		x_{i,s}^P\in\{0,1\},
		&&
		\forall i\in\mathcal I,\ 
		\forall s\in\mathbb T_i^P,
		\label{eq:ilp_prefill_binary}\\
		&
		x_{i,s}^D\in\{0,1\},
		&&
		\forall i\in\mathcal I,\ 
		\forall s\in\mathbb T_i^D.
		\label{eq:ilp_decode_binary}
	\end{align}
\end{subequations}

The objective minimizes the sum of request completion times.
Constraints~(\plaineqref{eq:ilp_unique_prefill}) and
(\plaineqref{eq:ilp_unique_decode}) assign one admission time to
each phase, while Constraint~(\plaineqref{eq:ilp_precedence}) enforces
prefill-to-decode precedence. In
Constraint~(\plaineqref{eq:ilp_prefill_capacity}), the difference between
the two cumulative sums equals one exactly when
\(S_i^P\leq t<S_i^D\), thereby enforcing the prefill memory capacity.
Constraint~(\plaineqref{eq:ilp_decode_capacity}) accounts for the
\(p_i+t-s+1\) units of decode memory used during slot \([t,t+1)\) by a
request admitted to decoding at time \(s\).

\begin{theorem}
	\label{thm:ilp_exactness}
	Problems~\(\mathbf{(ILP)}\) and~\(\mathbf{(O)}\) have the same optimal
	objective value.
\end{theorem}

\begin{proof}
	Consider an optimal schedule for \(\mathbf{(O)}\) with makespan at most
	\(H\). The horizon bound and precedence relation imply
	\(S_i^P\in\mathbb T_i^P\) and \(S_i^D\in\mathbb T_i^D\). Setting
	\(x_{i,S_i^P}^P=x_{i,S_i^D}^D=1\) and all other variables to zero yields
	a feasible solution to \(\mathbf{(ILP)}\) with the same objective value.
	
	Conversely, Constraints~(\plaineqref{eq:ilp_unique_prefill}) and
	(\plaineqref{eq:ilp_unique_decode}) define admission times \(S_i^P\) and
	\(S_i^D\) for every request. The remaining constraints enforce precedence
	and both memory capacities in every slot, while the objective equals
	\(\sum_i(S_i^D+d_i)=\sum_i C_i\). Thus, every feasible ILP solution defines
	a feasible schedule for \(\mathbf{(O)}\) with the same objective value.
\end{proof}

\section{Additional Experimental Results}
\label{app:experiments}

This appendix provides the complete CPU and GPU results supporting
Section~\ref{sec:experiments}. In the CPU experiments, ALG/OPT denotes the
algorithmic objective divided by the exact optimum obtained from the
time-indexed ILP, whereas ALG/LP uses the LP lower bound
\(\mathrm{LP}^{\star}\). In the GPU experiments, total completion time (TCT)
is normalized by the corresponding P/D baseline under the same
model, workload, and operating condition. 

\subsection{CPU Results}
\label{app:cpu_experiments}
We evaluate both algorithms on synthetic and trace-derived instances.
Each synthetic setting uses 30 random seeds. For each request, we
randomly select one of four sampling schemes with equal probability:
uniform, bimodal, heavy-tailed, or correlated sampling.
The correlated scheme generates \(m_i,p_i,d_i\) from a shared random
variable. All three parameters lie in \(\{1,\ldots,20\}\).
The supplementary code specifies the sampling rules.

For each trace workload, we use 30 windows of 250 requests and take
the first \(n\) requests from each window. We set \(p_i\) and \(d_i\)
to the recorded input and output token counts, respectively, and
use \(m_i=\lceil p_i/512\rceil\).
All CPU instances assume that requests are released at time zero.

For the capacity experiments, we set
\(M_P=\lceil\kappa\max_i p_i\rceil\) and
\(M_D=\lceil\kappa\max_i h_i\rceil\), where \(\kappa\) is the capacity
multiplier.

\paragraph{Synthetic Workloads.}
Figure~\ref{fig:cpu_exact} compares both algorithms with the exact ILP
optimum on small instances. ALG/OPT generally increases with \(n\), and
Algorithm~\ref{alg:bounded_scheduling} has the lower ratio across the tested
capacities. For larger instances, Figure~\ref{fig:cpu_lp_scaling} shows that
ALG/LP decreases as \(n\) grows, while LP solution time increases.

\begin{figure}[htp]
	\centering
	\includegraphics[width=0.90\linewidth]{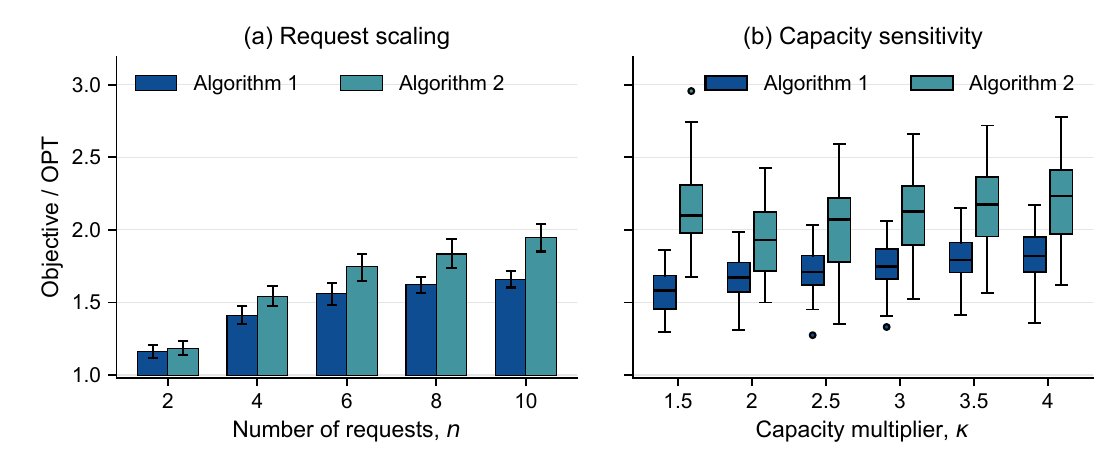}
	\caption{(a) ALG/OPT versus \(n\) at \(\kappa=2\).
		(b) ALG/OPT versus \(\kappa\) at \(n=10\).}
	\label{fig:cpu_exact}
\end{figure}

\begin{figure}[htp]
	\centering
	\includegraphics[width=0.90\linewidth]{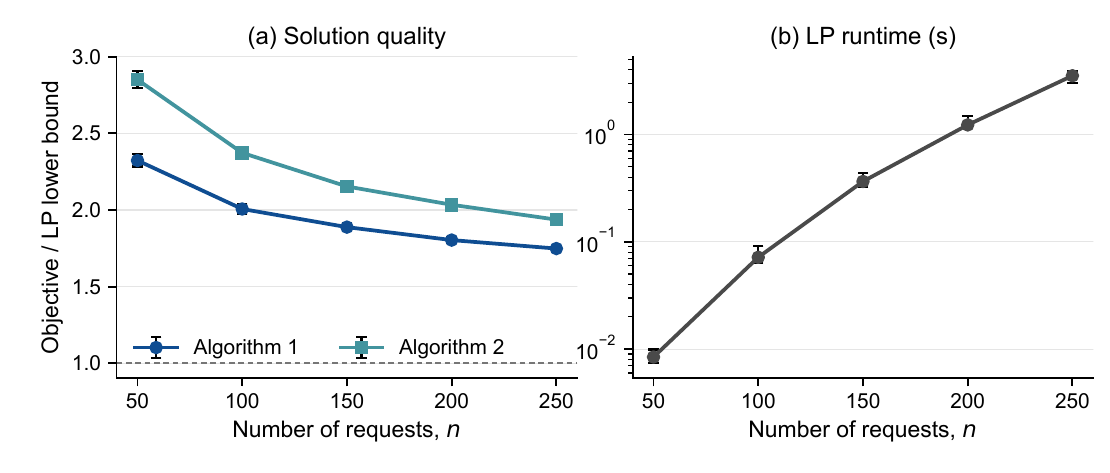}
	\caption{(a) ALG/LP versus \(n\) at \(\kappa=2\).
		(b) LP solution time versus \(n\) at \(\kappa=2\).}
	\label{fig:cpu_lp_scaling}
\end{figure}

\begin{figure}[htp]
	\centering
	\includegraphics[width=0.90\linewidth]{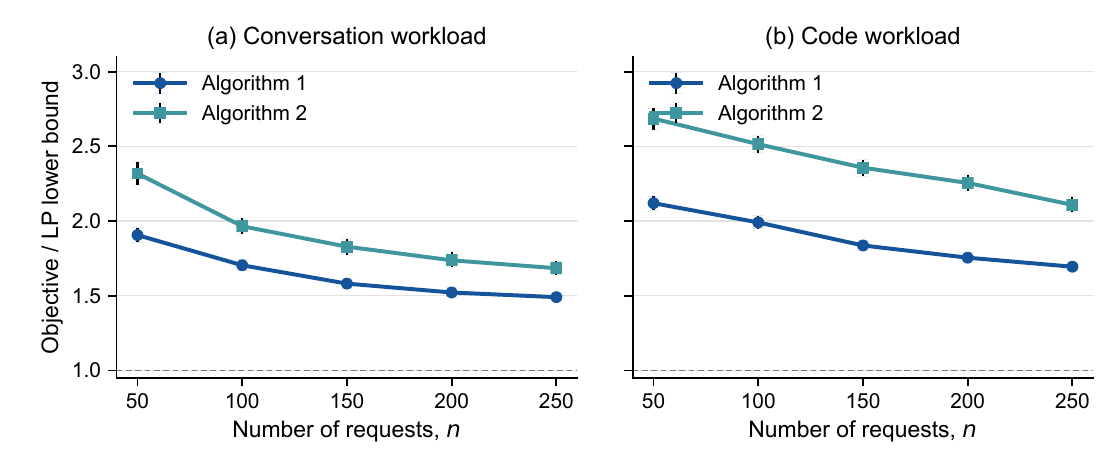}
	\caption{ALG/LP versus \(n\) at \(\kappa=2\) for
		(a) conversation and (b) code workloads.}
	\label{fig:cpu_trace_scaling}
\end{figure}

\paragraph{Trace-Derived Workloads.}
We use conversation and code workloads derived from the Azure LLM Inference
Trace 2023~\citep{splitwise}. Figure~\ref{fig:cpu_trace_scaling} shows that
ALG/LP decreases with \(n\) for both workloads, with
Algorithm~\ref{alg:bounded_scheduling} yielding lower ratios throughout.

\begin{figure}[htp]
	\centering
	\includegraphics[width=0.90\linewidth]{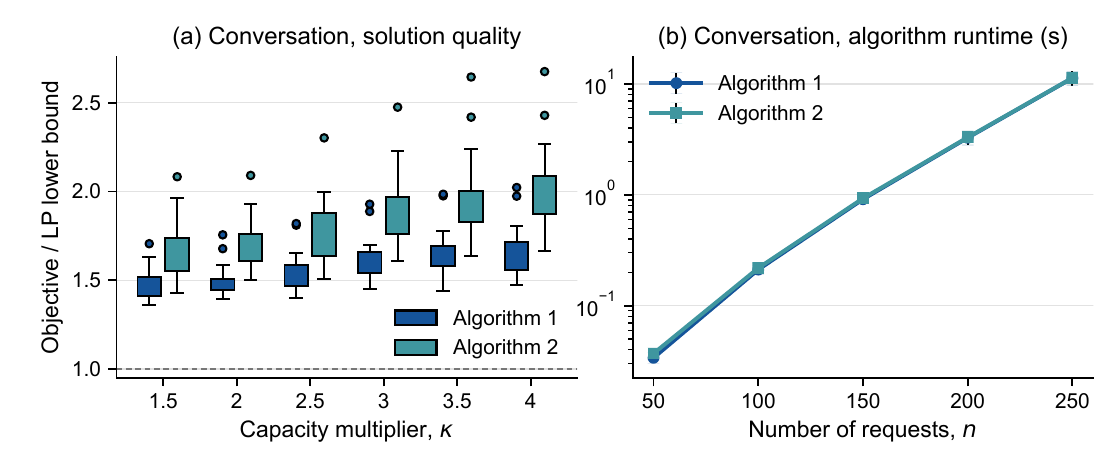}
	\caption{Conversation workload. (a) ALG/LP versus \(\kappa\) at \(n=250\).
		(b)  Runtime versus \(n\) at \(\kappa=2\).}
	\label{fig:cpu_conversation_capacity}
\end{figure}

\begin{figure}[htp]
	\centering
	\includegraphics[width=0.90\linewidth]{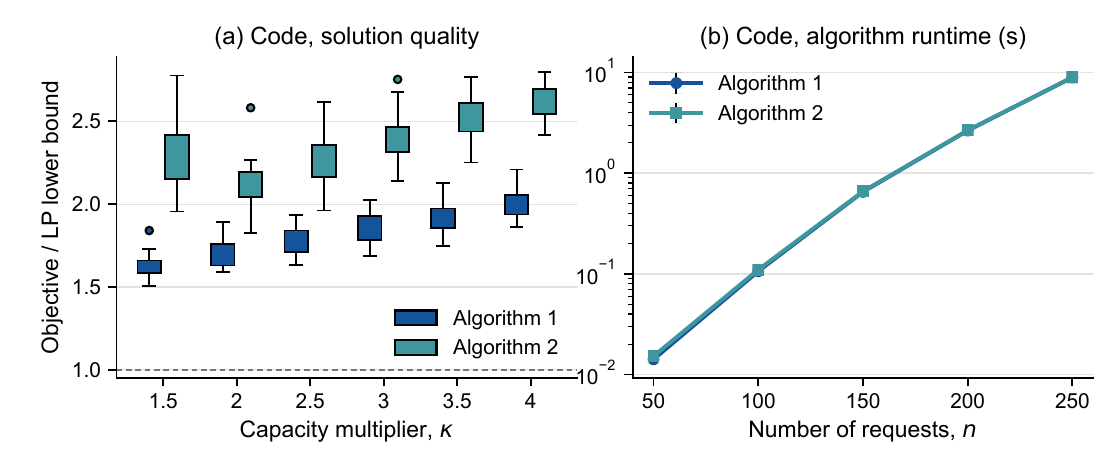}
	\caption{Code workload. (a) ALG/LP versus \(\kappa\) at \(n=250\).
		(b) Runtime versus \(n\) at \(\kappa=2\).}
	\label{fig:cpu_code_capacity}
\end{figure}

\begin{figure}[htp]
	\centering
	\includegraphics[width=0.90\linewidth]{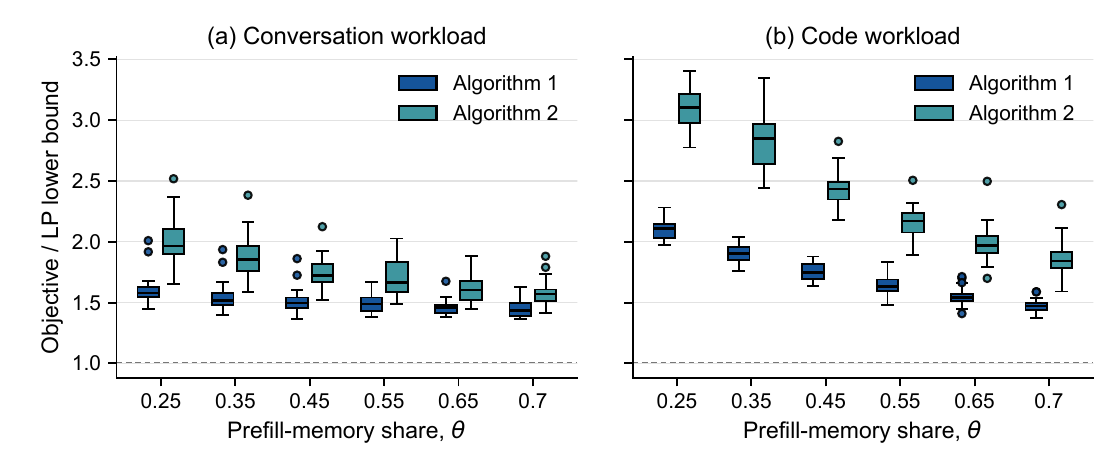}
	\caption{ALG/LP versus \(\theta\) for
		(a) conversation and (b) code workloads.}
	\label{fig:cpu_pool_split_ratio}
\end{figure}

Figures~\ref{fig:cpu_conversation_capacity}
and~\ref{fig:cpu_code_capacity} show that ALG/LP generally increases with
\(\kappa\), while algorithm runtime increases with \(n\). Runtime includes
both LP solution and schedule construction. The two runtime curves nearly
overlap because both algorithms use the same LP solution, which accounts for
most of their runtime, especially for larger \(n\).

We further examine how a fixed total memory budget is divided between the
two pools. For each instance, let
\(B=\lceil2\max_i p_i\rceil+\lceil2\max_i h_i\rceil\)
denote the total budget at \(\kappa=2\). We set
\(M_P=\lfloor\theta B\rfloor\) and \(M_D=B-M_P\), and vary
\(\theta\in\{0.25,0.35,0.45,0.55,0.65,0.70\}\) at \(n=250\).

Figures~\ref{fig:cpu_pool_split_ratio}
and~\ref{fig:cpu_pool_split_objective} report ALG/LP and total completion
time (TCT), respectively. ALG/LP generally decreases as \(\theta\) increases.
By contrast, TCT rises with \(\theta\) for the conversation
workload and is lowest near \(\theta=0.35\) for the code workload.

\begin{figure}[htp]
	\centering
	\includegraphics[width=0.90\linewidth]{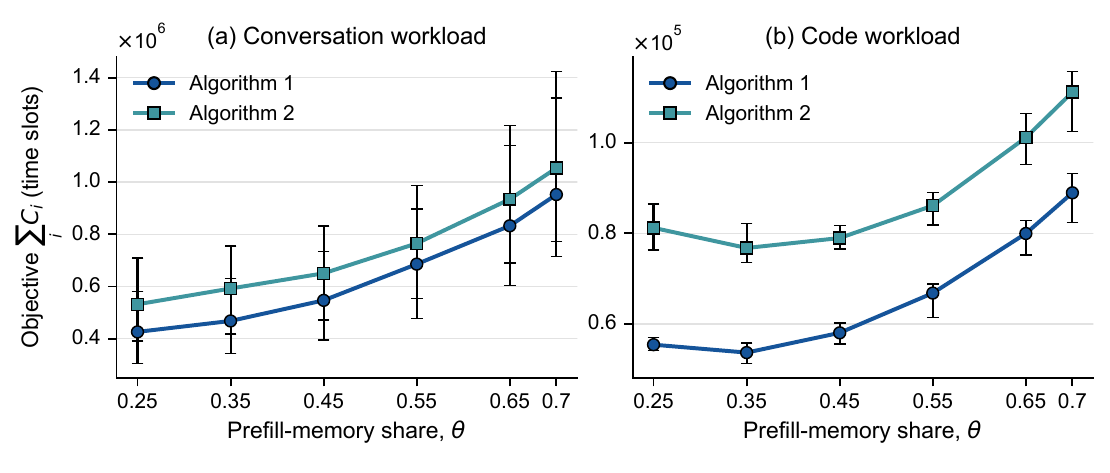}
	\caption{Total completion time versus \(\theta\) for
		(a) conversation and (b) code workloads.}
	\label{fig:cpu_pool_split_objective}
\end{figure}

\subsection{GPU Serving Results}
\label{app:gpu_experiments}
We conduct serving experiments on two NVIDIA H100 GPUs using vLLM~0.12.0
with BF16 precision. We evaluate Qwen2.5
(7B)~\citep{yang2024qwen25} and Llama~3.1
(8B)~\citep{grattafiori2024llama} on conversation and code workloads.

For each request \(i\), \(p_i\) is the input token count measured
by the model's tokenizer. We fix the output length to \(d_i\) tokens
by ignoring the end of sequence token.
We set \(m_i=\lceil \widehat{t}_i^P/\tau\rceil\), where
\(\widehat{t}_i^P\) is the prefill time estimated by interpolating
measured prefill times at different prompt lengths on logarithmic
scales, and \(\tau\) is the measured median time per decode step
for the model.
For \(n=250\), \((\max_i p_i,\max_i h_i)\) equals
\((4107,4176)\) for conversation and \((7436,7448)\) for code
for both models.

P/D and Algorithms~\ref{alg:bounded_scheduling}
and~\ref{alg:general_scheduling} use one GPU for prefill and one for
decoding. Colocated runs a single tensor parallel vLLM instance across
both GPUs, which execute both phases. Colocated Chunked uses the same
configuration with chunked prefill enabled. Each KV cache block holds 16 tokens. For P/D and both algorithms,
we configure \(\lceil M_P/16\rceil\) prefill blocks and
\(\lceil M_D/16\rceil\) decode blocks. Both colocated baselines use
a unified cache with \(\lceil M_D/16\rceil\) blocks.

\paragraph{Admission Control.}
An external proxy enforces the offline schedule before forwarding requests
to the two vLLM instances. It waits until the full batch is registered,
then applies the algorithm's group barriers and admission order. To enforce
the prefill capacity, the proxy counts both requests executing prefill and
requests retaining their prompt cache while awaiting decode. Before forwarding
request \(i\) to decode, it checks
\(\sum_{j\in\mathcal D}h_j+h_i\leq M_D\), where \(\mathcal D\) contains
requests admitted to decode but not yet completed. When request \(i\) completes,
the proxy releases \(h_i\) from this accounting budget. This is a software
admission rule, not a physical reservation of \(h_i\) KV cache units inside
vLLM, which manages actual KV allocation separately.

\paragraph{KV Transfer.}
The prefill and decode instances transfer KV state using vLLM's
P2pNcclConnector. The recorded GPU topology shows an NV12 connection
between the two H100 GPUs. Although the theoretical model assumes
instantaneous KV transfer, the completion times measured in the GPU
experiments include KV transfer overhead.

\paragraph{Memory Capacity.}
We use \(n=250\) requests and vary
\(\kappa\in\{1.5,2,2.5,3,3.5,4\}\), where
\(M_P=\lceil\kappa\max_i p_i\rceil\) and
\(M_D=\lceil\kappa\max_i h_i\rceil\).

Figure~\ref{fig:gpu_capacity_tct} shows that both algorithms reduce TCT
relative to P/D across all tested capacities on the code workloads. On the
conversation workloads, their TCT advantage diminishes as capacity increases.

\begin{figure}[htp]
	\centering
	\includegraphics[width=0.90\linewidth]{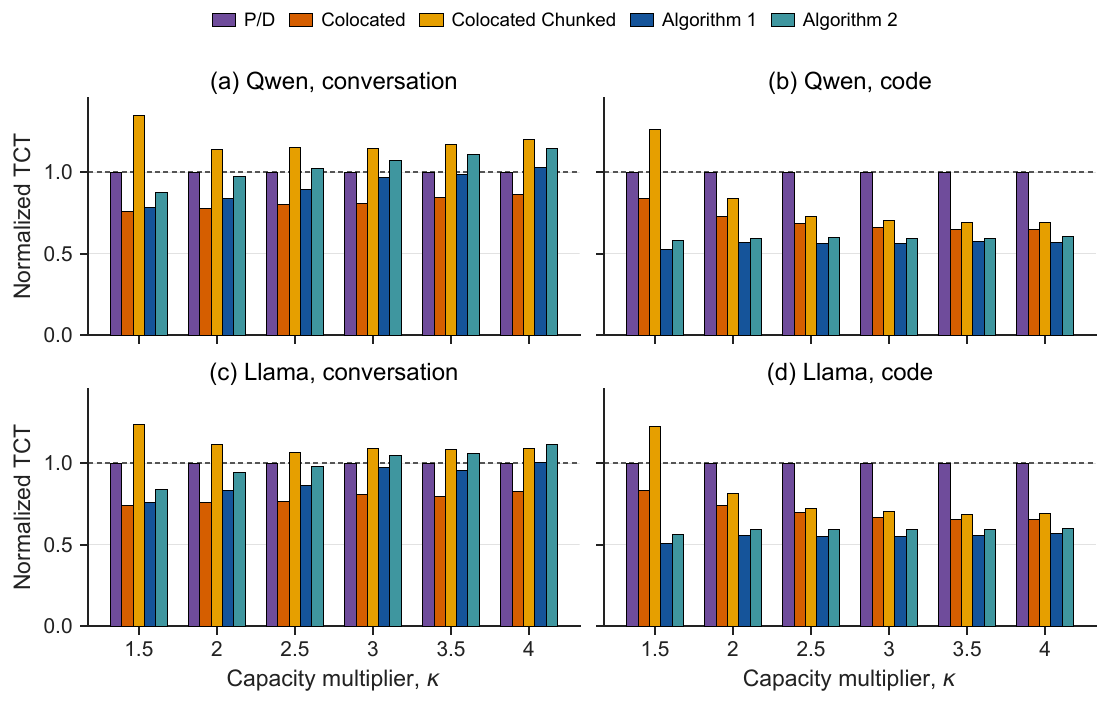}
	\caption{Total completion time as \(\kappa\) varies, normalized by the P/D result.}
	\label{fig:gpu_capacity_tct}
\end{figure}

At \(\kappa=2\), Figure~\ref{fig:gpu_capacity_tail} shows that both
algorithms have lower P50 but higher P99 request completion times than
P/D across all four model and workload combinations. Their P95 times
are lower than P/D on code workloads but higher on conversation workloads.

\begin{figure}[htp]
	\centering
	\includegraphics[width=0.90\linewidth]{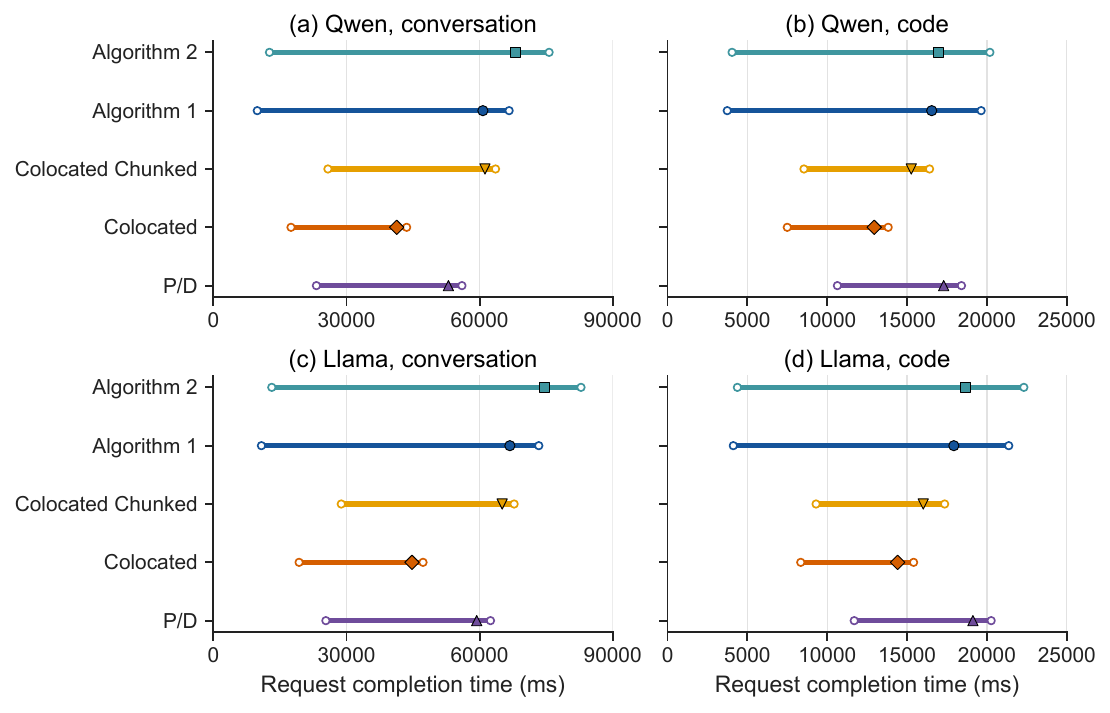}
	\caption{Request completion times at \(\kappa=2\) and \(n=250\).
		Each segment spans P50 to P99; the filled marker denotes P95.}
	\label{fig:gpu_capacity_tail}
\end{figure}

\paragraph{Prefill and Decode Demands.}
We use \(n=250\) requests and target three demand settings:
\((\alpha,\beta)=(0.9,0.5)\), \((0.75,0.75)\), and \((0.5,0.9)\),
representing prefill-heavy, balanced, and decode-heavy demands,
respectively.

Figure~\ref{fig:gpu_pressure_tct} shows that, for both models, both
algorithms reduce TCT relative to P/D on code workloads under all three
demand settings. The reductions are generally larger than those on
conversation workloads. On conversation workloads,
Algorithm~\ref{alg:bounded_scheduling} remains below P/D, whereas
Algorithm~\ref{alg:general_scheduling} slightly exceeds P/D for Qwen
under prefill-heavy demand.

\begin{figure}[H]
	\centering
	\includegraphics[width=0.90\linewidth]{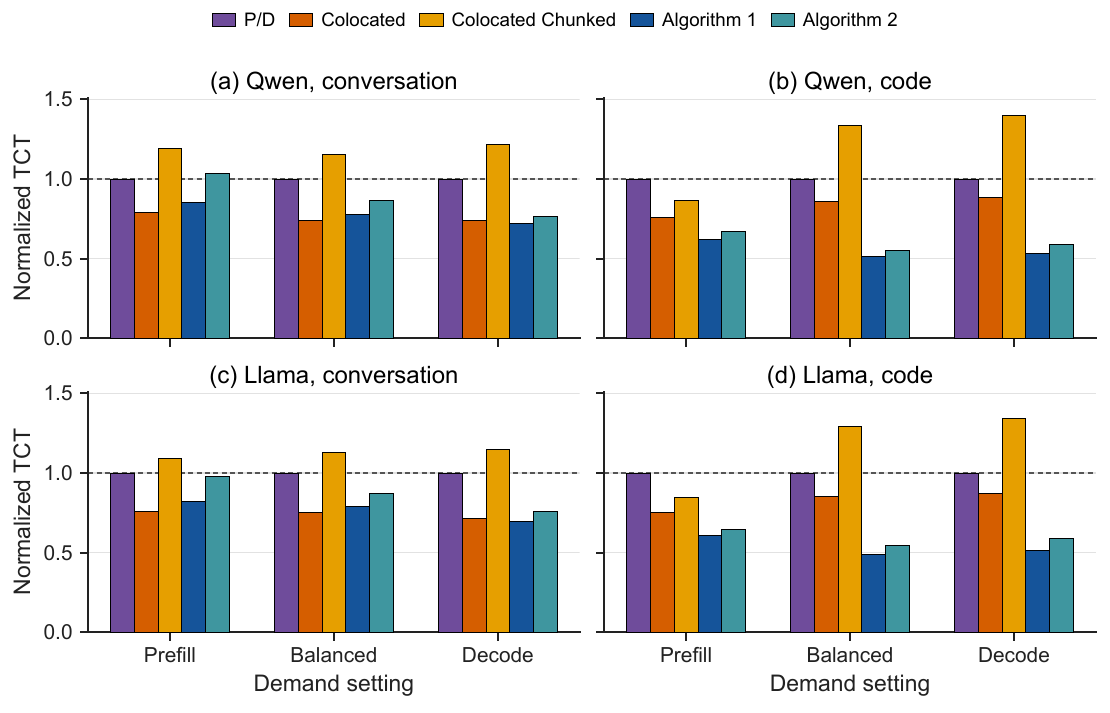}
	\caption{Total completion time for different values of \(\alpha\) and
		\(\beta\), normalized by the P/D result.}
	\label{fig:gpu_pressure_tct}
\end{figure}

Time to first token (TTFT) is measured from the time each request is sent
until its first output token is received. In the balanced setting,
Figure~\ref{fig:gpu_pressure_tail} shows that both algorithms have lower
P50 TTFT than P/D for both models and workloads. Their P99 TTFT is higher
than P/D, except for Algorithm~\ref{alg:bounded_scheduling} on the Llama
code workload, where it is slightly lower.

\begin{figure}[htp]
	\centering
	\includegraphics[width=0.90\linewidth]{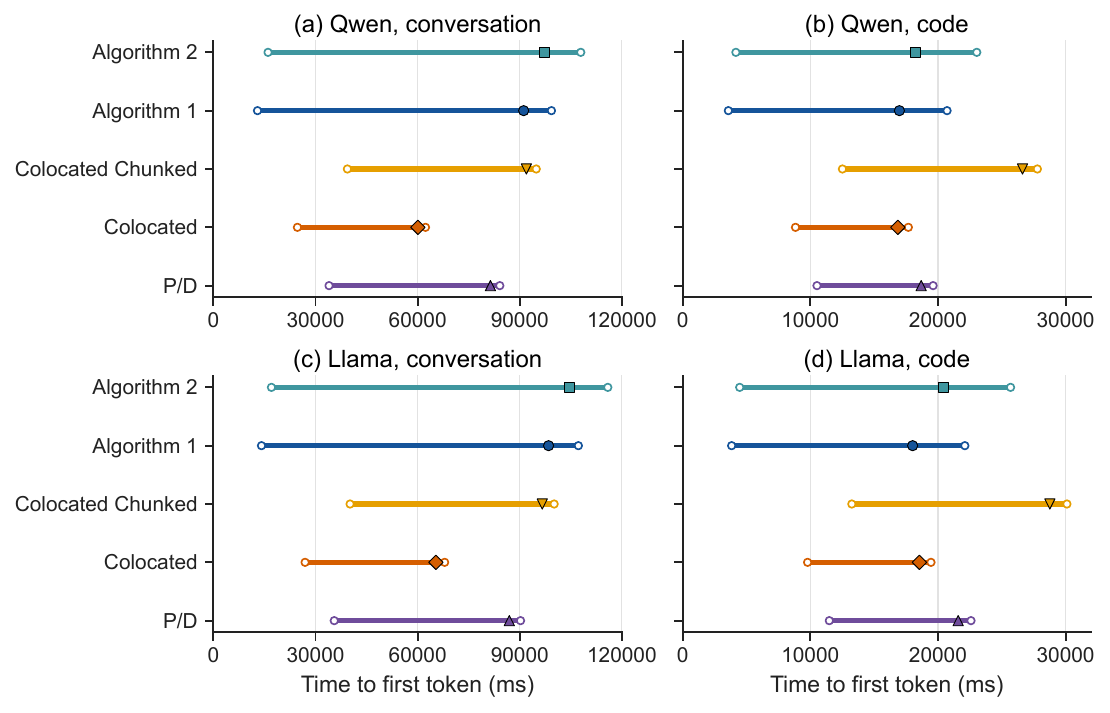}
	\caption{TTFT at \((\alpha,\beta)=(0.75,0.75)\).
		Each segment spans P50 to P99; the filled marker denotes P95.}
	\label{fig:gpu_pressure_tail}
\end{figure}

\paragraph{Number of Requests.}
We vary \(n\in\{50,100,150,200,250\}\). The capacities are fixed across \(n\) at
\((M_P,M_D)=(8214,8352)\) for conversation and
\((14872,14896)\) for code, corresponding to the \(n=250\),
\(\kappa=2\) settings.

Figure~\ref{fig:gpu_scaling_tct} shows that the algorithms' normalized
TCT generally decreases as \(n\) increases. Both algorithms remain
below P/D across all batch sizes on code workloads; on conversation
workloads, their TCT advantage appears at larger batch sizes.

\begin{figure}[htp]
	\centering
	\includegraphics[width=0.90\linewidth]{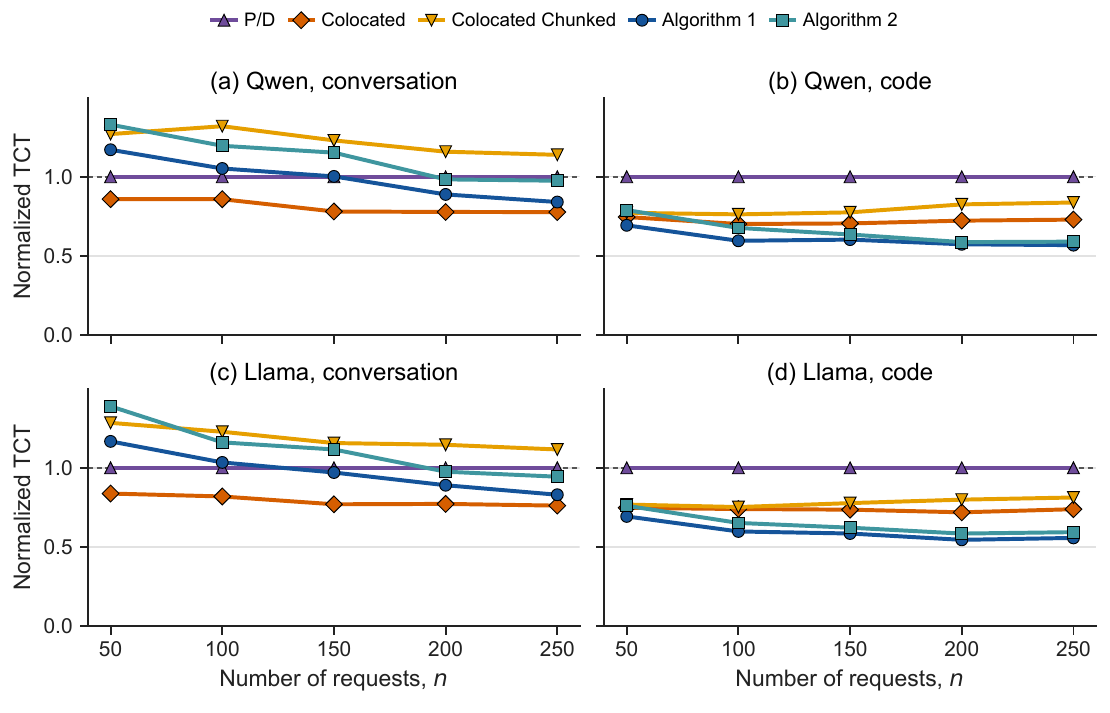}
	\caption{Total completion time as \(n\) varies, normalized by the
		P/D result.}
	\label{fig:gpu_scaling_tct}
\end{figure}

\begin{figure}[H]
	\centering
	\includegraphics[width=0.90\linewidth]{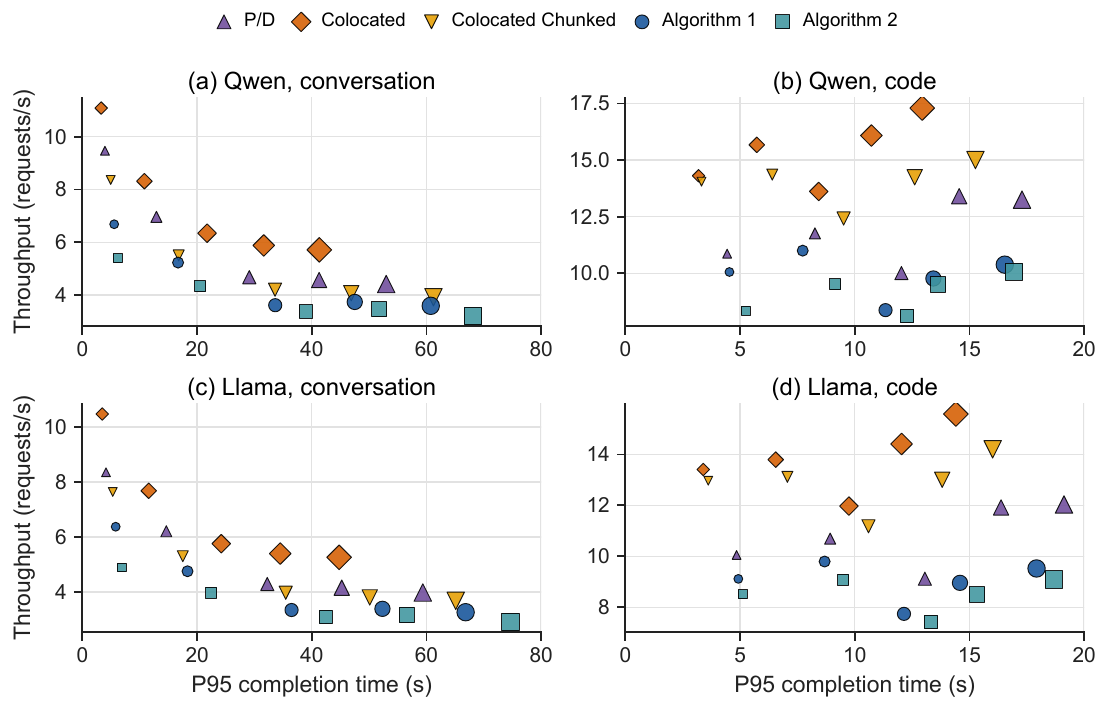}
	\caption{Throughput versus P95 request completion time as \(n\) varies.
		Marker size increases with \(n\).}
	\label{fig:gpu_scaling_tradeoff}
\end{figure}

Figure~\ref{fig:gpu_scaling_tradeoff} compares throughput and P95 request
completion time at each batch size. At comparable values of \(n\),
Colocated achieves higher throughput and lower P95 completion time than
either algorithm.

\begin{figure}[htp]
	\centering
	\includegraphics[width=0.90\linewidth]{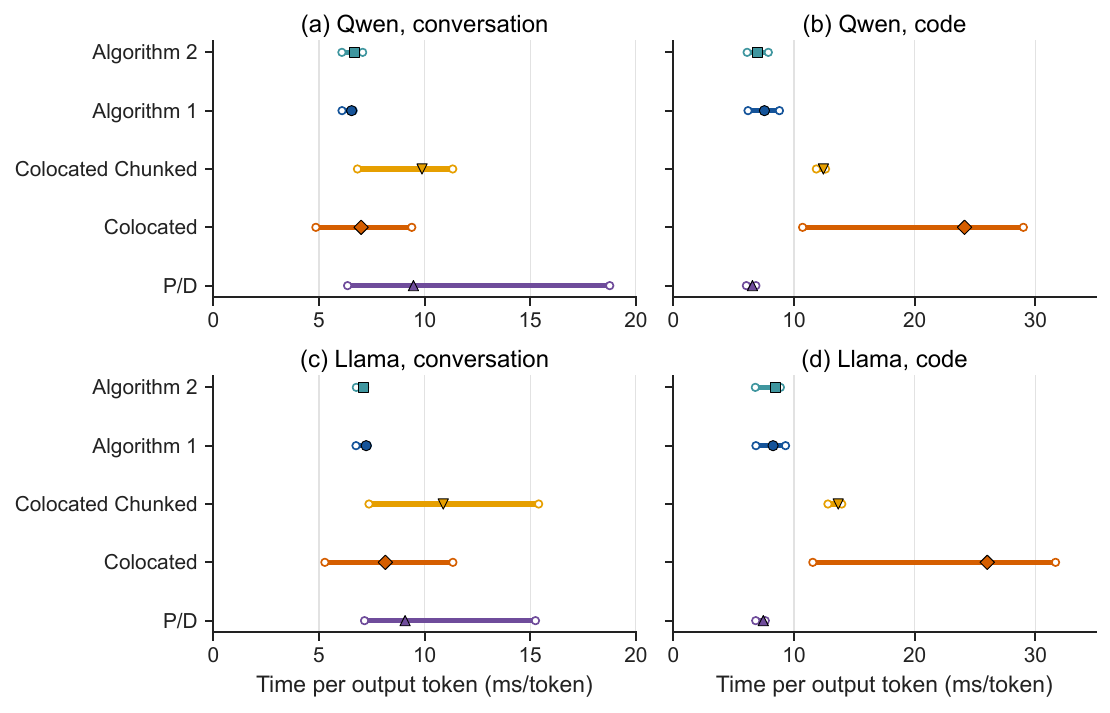}
	\caption{TPOT at \(n=250\). Each segment spans P50 to P99;
		the filled marker denotes P95.}
	\label{fig:gpu_scaling_tpot}
\end{figure}

Time per output token (TPOT) measures the average generation time per output
token after the first token. At \(n=250\),
Figure~\ref{fig:gpu_scaling_tpot} shows that both algorithms have lower P95
TPOT than P/D on the conversation workloads and slightly higher P95 TPOT
than P/D on the code workloads. Both algorithms remain below Colocated
Chunked, while Colocated has substantially higher P95 TPOT on the code
workloads.

\end{document}